\documentclass[a4paper,UKenglish]{lipics-v2018}

\usepackage{packages/mymaths}
\usepackage{packages/mylambdas}
\usepackage{packages/proofnets}
\usepackage{bbm}
\usepackage[toc,page]{appendix}
\usepackage{graphicx}
\usetikzlibrary{matrix}
\usepackage{microtype}

\newcolumntype{L}{>{$}l<{$}}
\newcolumntype{R}{>{$}r<{$}}
\newcolumntype{C}{>{$}c<{$}}
\newcolumntype{P}[1]{>{\hfill}p{#1}<{\hfill\vphantom.}}
\newcolumntype{h}[1]{@{\hspace{#1}}}
\newcolumntype{b}{@{}>{{}}}

\newcolumntype{D}{ >{\centering\arraybackslash} m{5cm}}
\newcolumntype{E}{ >{\centering\arraybackslash} m{1cm}}

\newcommand{\mynl}{\\[1.5ex]}

\newtheorem{notation}{Notation}
\newtheorem{property}{Property}

\newcommand{\store}[2]{{#1} \Leftarrow {#2}}

\newcommand{\setw}{assignment}
\newcommand{\getw}{read}

\title{Interpreting a concurrent $\lambda$-calculus in differential proof nets -
Extended version}

\author{Yann Hamdaoui}{IRIF, Univ. Paris Diderot}{yann.hamdaoui@irif.fr}{}{}
\authorrunning{Y. Hamdaoui}
\Copyright{Yann Hamdaoui}

\subjclass{Dummy classification}


\EventEditors{John Q. Open and Joan R. Access}
\EventNoEds{2}
\EventLongTitle{42nd Conference on Very Important Topics (CVIT 2016)}
\EventShortTitle{CVIT 2016}
\EventAcronym{CVIT}
\EventYear{2016}
\EventDate{December 24--27, 2016}
\EventLocation{Little Whinging, United Kingdom}
\EventLogo{}
\SeriesVolume{42}
\ArticleNo{23}
\begin{document}

\maketitle

\begin{abstract}
In this paper, we show how to interpret a language featuring concurrency,
references and replication into proof nets, which correspond to a fragment of
differential linear logic. We prove a simulation and adequacy
theorem. A key element in our translation are routing areas, a family of
nets used to implement communication primitives which we define and
study in detail.
\end{abstract}

\section{Introduction}
\label{sec:intro}

\newcommand{\LL}{\textsf{LL}}
\newcommand{\torework}[1]{{\color{red}{#1}}}

% Plan :
% \begin{itemize}
%   \item Linear Logic : what ? (Logic of resource, decompose, fine grained)
%  \item Translations into Linear logic : generic information, interesting, semantics,
%  etc.
%  \item But also from a concrete perspective : parallelization, GoI, complexity,
%  compilation, linear types
%  \item We go a step further to encode concurrency and references into proof nets,
%  constructing on [erhardlaurent] and [tranquilli]. Amadio calculus (examples ?)
%  \item Doing so, we introduct a generalization of communication areas, routing areas
%  little blahblah and a language with references explicit substitutions, modeling
%  the message-passing dynamic of proof nets.
% \end{itemize}
The distinctive feature of Linear Logic~\cite{GIRARD19871} (\LL) is to be a
resource-aware logic, which explains its success as a tool to study
computational processes. The native language for {\LL} proofs - or rather,
\emph{programs} - are proof nets, a graph representation endowed with a local
and asynchronous cut-elimination procedure.  The fine-grained computations and
the explicit management of resources in {\LL} make it an expressive target to
translate various computational primitives. Girard provided two translations in
its original paper, later clarified as respectively a call-by-value and
call-by-name translation of the $\lambda$-calculus~\cite{MARAIST1995370}.
PCF~\cite{Mackie:1995:GIM:199448.199483} has also been considered. Other
works have tackled the intricate question of modeling side-effects. State has been considered in a $\lambda$-calculus with
references~\cite{tranquilli:hal-00465793}. Another direction which has been explored is
concurrency and non-determinism. The extension of {\LL} with dual structural
rules, Differential Linear Logic~\cite{ehrhard:hal-00150274}, happens to be
powerful enough to accommodate non-determinism as demonstrated by the encoding
of a $\pi$-calculus without replication~\cite{EHRHARD2010606}.

These translations allow to relate and compare different computational
frameworks.  They benefit from the consequent work on the semantic and the
dynamic of {\LL} which has been carried for thirty years.  They are also of
practical interest : the proof net representation naturally leads to parallel
implementations~\cite{mackie94,Pedicini:2007,Pinto:2001} and forms the basis for
concrete operational semantics in the form of token-based
automata~\cite{DANOS199640,Lago:2015:PSI:2876514.2876564}. In this regard, one
may consider {\LL} as a ``functional assembly language'' with a diverse collection of
semantics and abstract machines.

While $\lambda$-calculus is a fundamental tool in the study and design of
functional programming languages, mainstream programming languages are pervaded
with features that enable productive development such as support for parallelism
and concurrency, communication primitives, imperative references, \emph{etc }.
Most of them imply side-effects, which are challenging to model, to compose, and
to reason about. While some have been investigated individually through the lens
of {\LL}, our goal is to go one step further by modeling a language featuring at
the same time concurrency, references and replication. The constructs involved in the
translation are inspired by the approach in~\cite{EHRHARD2010606} for
concurrency and non-determinism coupled with a monadic translation~\cite{tranquilli:hal-00465793} for
references.  Our goal is to exploit the ability of proof nets to enable
independent computations to be done in parallel without breaking the original operational
semantic. The translation we propose can be seen as a compilation
from a {\em global shared memory} model to a {\em local message passing} one, in
line with proof nets philosophy.

\paragraph*{Routing areas}

In a concurrent imperative language, references are a means to exchange
information between threads. In order to implement communication primitives in
proof nets, we use and extend the concept of \emph{communication areas}
introduced in \cite{EHRHARD2010606}. A \emph{communication area} is a particular
proof net whose external interface, composed of wires, is split between an equal
number of inputs and outputs. Inputs and outputs are grouped by pairs
representing a plug on which other nets can be connected. They are simple yet
elegant devices, whose role is similar to the one of a network switch which
connects several agents.  Connecting two \emph{communication areas} yields a
\emph{communication area} again: this key feature enables their use as modular
blocks that can be combined into complex assemblies. In this paper, we introduce
\emph{routing areas}, which allows a finer control on the wiring diagram. They
are parametrized by a relation which specifies which inputs and outputs are
connected.  Extending our network analogy, \emph{communication areas} are rather
hubs: they simply broadcast every incoming message to any connected agent. On
the other hand, \emph{routing areas} are more like switches: they are able to
choose selectively the recipients of messages depending on their origin.
\emph{Routing areas} are subject to atomic operations that decompose the
operation of connecting \emph{communication areas}. These operations also
have counterparts on relations.

We show that \emph{routing areas} are sufficient to actually describe all the normal
forms of the fragment of proof nets composed solely of structural
rules.  The algebraic description of \emph{routing areas} then provides a
semantic for this fragment.

\paragraph*{A Concurrent $\lambda$-calculus}
We consider a paradigmatic concurrent lambda-calculus with higher-order
references -- {\lamadio} below -- which has been introduced by Amadio
in~\cite{Amadio2009}. It is a call-by-value $\lambda$-calculus extended with:
\begin{itemize}
\item a notion of threads and an operator $\parallel$ for parallel
  composition of threads,
\item two terms $\set{r}{V}$ and $\get{r}$, to respectively
  assign a value to and read from a reference,
\item special threads $\store{r}{V}$, called stores, accounting for
  assignments.
\end{itemize}
%
% When $\set{r}{V}$ is reduced, it turns to the unit value $\ast$ and produces a
% store $\store{r}{V}$ making the value available to all the other threads. A
% corresponding construct $\get{r}$ is reduced by choosing non deterministically a
% value among all the available stores. For example, assuming some support for
% basic arithmetics consider the program
% $
% (\lambda x. x + 1)\ \get{r} \parallel
%   \set{r}{0} \parallel \set{r}{1}.
% $
% It consists of 3 threads: two concurrent assignments $\set{r}{0}$
% and $\set{r}{1}$, and an application $(\lambda x. x + 1)\,\get{r}$.
% This programs admits two normal forms depending on which
% assignment ``wins'':  the term $1 \parallel
% \ast \parallel \ast \parallel r \Leftarrow 0 \parallel
% \store{r}{1}$ and the term $2 \parallel
% \ast \parallel \ast \parallel r \Leftarrow 0 \parallel
% \store{r}{1}$.
%
In this language, the stores are global and cumulative: their scope is the whole
program, and each {\setw} adds a new binding that does not erase the previous
ones. Reading from a store is a non deterministic process that chooses a value
among the available ones. References are able to handle an unlimited number of
values and are understood as a typed abstraction of possibly several concrete
memory cells. This feature allows {\lamadio} to simulate various other calculi
with references such as variants with dynamic references or
communication~\cite{madet:tel-00794977}.  The language is endowed with an
appropriate type and effects system ensuring termination. The translation is
presented on an explicit substitutions version of {\lamadio}, introduced to
serve as an intermediate representation. This intermediate language is presented
and studied in details in \cite{HamdaouiValiron2018}.

\paragraph*{Contributions.}
The contributions of this paper are:
\begin{itemize}
  \item We introduce and study \emph{routing areas} which are flexible devices
    for implementing communication in proof nets. From routing areas we derive a
    semantic for a fragment of proof nets.
  \item We illustrate the use of routing areas by translating a concurrent
    $\lambda$-calculus to proof nets. Routing areas are the building block of
    this translation.
  \item We prove a \emph{simulation} and \emph{adequacy} theorem for this
    translation.
\end{itemize}

\paragraph*{Plan of the paper.}
\begin{itemize}
  \item In Section~\ref{section:nets}, we detail our proof nets setting and state its 
    properties
  \item We go on to define routing areas and study them more in detail
    in Section~\ref{section:areas}.
  \item In Section~\ref{section:lamadio}, we introduce the source language of our
    translation, the concurrent $\lambda$-calculus \lamadio.
  \item Section~\ref{section:translation} is devoted to the translation. We explain its underlying
    principles and give some representative cases.
  \item Section~\ref{section:properties} gives the simulation, termination and
    adequacy theorems together with their proof.
\end{itemize}

\section{Proof nets}
\label{section:nets}

In this section, we detail the proof nets that we are considering in the rest of
the paper. We do not attempt to give a full treatment of proof nets.
We recall the important notions and specify the system that we use. The
interested reader may find more details in~\cite{Pag08} for example. Proof nets are a
representation of proofs as multigraphs, where edges - called \emph{wires}
- correspond to the formulas, and nodes - \emph{cells} - correspond to the
rules.  We can decompose our system into three different layers:
\begin{description}
  \item[Multiplicative] The multiplicative fragment is
    composed of the conjunction $\otimes$ and the dual disjunction $\parr$.
    These connectors can express the linear implication $\multimap$ and are thus
    able to encode a linear $\lambda$-calculus, where all bound variables
    must occur exactly once in the body of an abstraction.
  \item[Exponential] The exponentials enable
    structural rules to be applied on particular formulas distinguished by the
    $\oc$ modality (its dual being $\wn$). Structural rules correspond to
    duplication (contraction) and erasure (weakening) : the multiplicative
    exponential fragment is the typical setting to interpret the
    $\lambda$-calculus.
  \item[Differential] Non determinism is expressed by using two rules from
    Differential {\LL}: cocontraction and coweakening.
    Semantically, contraction is thought as a family of diagonal morphisms
    $\mathit{cntr}_A : \oc A \to \oc A \otimes \oc A$, each one taking a
    resource $\oc A$ and duplicate it into a pair $\oc A \otimes \oc A$.
    Dually, cocontraction is a morphism going in the opposite direction, packing
    two resources of the same type into one : $\mathit{cocntr}_A : \oc A \otimes
    \oc A \to \oc A$.  What happens when the resulting resource is to be
    consumed ? Two incompatible possibilities : either the left one is used and
    the right one is erased, or vice-versa.  This corresponds to the rule
    $\red{nd}$ (see Table~\ref{figure:reduction}) in proof nets: the reduction
    produces the non-deterministic sum of the two outcomes.  Cocontraction will
    be used as an internalized non-deterministic sum. While weakening
    $\mathit{weak}_A : \oc A \to \mathbf{1}$ erases a resource, the dual
    coweakening produces a resource ex nihilo: $\mathit{coweak}_A : \mathbf{1}
    \to \oc A$.  This resource can be duplicated or erased, but any attempt to
    consume it will turn the whole net to $\mathbf{0}$. Coweakening produces a
    Pandora box with a $0$ inside. It is the neutral element of cocontraction.
\end{description}

% We define in this section the proof nets we use as a target of the translation.
% Proof nets are a representation of proofs as multigraphs, where pending edges -
% called here \emph{free wires} - correspond to conclusions, and nodes correspond
% to rules. We work in an extension of Multiplicative Exponential Linear Logic
% (MELL), to which we add the coweakening and cocontraction rules borrowed to
% Differential Linear Logic to represent non determinism.
%

One can define a correctness criterion to discriminate nets that are
well-behaved - the ones that are the representation of a valid proof - ensuring
termination and confluence of the reduction. We will not require it here (cf
Section~\ref{section:conclusion}). Without the correctness criterion, the full fledged reduction
of (differential) {\LL} is not even confluent, let alone terminating. We add
constraints in order to recover an suitable system that is confluent and
verifies a termination property (Theorem~\ref{theorem-nets-termination}). Let us now give the
definition of proof nets and their reduction:
% As mentioned in the introduction, the translation of concurrent processes
% naturally generate cycles, thus no correctness criterion is involved. To cope
% with this permissive setting, we restrict the reduction not only to be surface
% (no reduction inside boxes), but also close, meaning that only a box with no
% auxiliary doors can be subject to exponential reductions. Actually, we
% liberalize a bit the surface constraints to allow a few harmless reduction rules
% to happen at arbitrary depth. We prove that the resulting system is
% \emph{confluent}, and satisfies a termination-related property: a weakly
% normalizing term is actually strongly normalizing.

\begin{notation}
We recall some vocabulary of rewriting theory that we use in the following. A term of a language or a proof net $t$ is
\begin{description}
  \item[A normal form] if it can't be reduced further
  \item[Weakly normalizing] if there exists a
    normal form $n$ such that $t \to^\ast n$ 
  \item[Strongly normalizing] if it has no infinite
    reduction sequence
  \item[Confluent] if for all $u,u'$ such that $u \rredstar{} t \to^\ast u'$,
    there exists $v$ such that $u \to^\ast v \rredstar{} u'$
  \end{description}
  A rewriting relation is confluent if all terms are.
\end{notation}

\begin{definition}{Proof nets}\label{nets-definition-nets}\\
Given a countable set, whose elements are called ports, a proof net is given by
\begin{enumerate}
  \item A finite set of ports
  \item A finite set of cells. A cell is a finite non-empty sequence of pairwise
    distinct ports, and two cells have pairwise distinct ports.  The first port
    of a cell $c$ is called the {\emph principal port} and written $p(c)$, and
    the $(i+1)$th the $i$th auxiliary port noted $p_i(c)$.  The number of
    auxiliary ports is called the {\emph arity} of the cell. A port is free if
    it does not occur in a cell. 
  \item A labelling of cells by symbols amongst $\{ 1, \otimes, \parr, \wn, \oc
    \}$. We ask moreover that the arity respects the following table: 
    \begin{center}
      \begin{tabular}{c|ccccc}
        Symbol & $\parr$ & $\otimes$ & $\wn$ & $\oc$ & $\oc_p$  \\
        \hline
        Arity &  $2$ & $2$ & $0,1$ or $2$ & $0,1$ or $2$ & $i \geq 1$\\
      \end{tabular}
    \end{center}
  \item A partition of its ports into pairs called {\emph wires}. A wire with
    one (resp. two) free port is a free (resp. floating) wire.
  \item A labelling of wires by MELL formulas: $F \mathrel{::=} 1 \mid \bot
    \mid F \parr F \mid F \otimes F \mid \oc F \mid \wn F$. A wire $(p_1,p_2)$
    labelled by $F$ is identified with the reversed wire $(p_2,p_1)$ labelled by
    $F^\perp$. We will abusively confuse ports with the wire attached to them.
  \item A mapping from $\oc_p$ nodes - called exponential boxes - to proof nets
    with no floating wires and $n \geq 1$ free wires. The first one is labelled
    (assuming orientation toward the free port) by $A$ and the
    remaining ones by $\oc B_1,\ldots, \oc B_{n-1}$. The corresponding
    $\oc_p$ node must have $\oc{A}$ as principal port and $\oc B_1,\ldots,\oc
    B_{n-1}$ as auxiliary ports.
\end{enumerate}
\end{definition}

The different kind of cells are illustrated in Figure~\ref{figure:nets}.  We
directly represent $\oc_p$ cells as their associated proof net delimited by a
rectangular shape. We impose that labels of wires connected to a cell respect
the one given in Figure~\ref{figure:nets}, i.e that nets are \emph{well-typed}.

%%% ----------------------------------------------------------------
\begin{table}[tb]
\centering
\begin{minipage}{.74\linewidth}
\centering
\begin{tabular}{cccc}
  one & tensor & par & dereliction \\
  \includegraphics{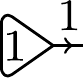} &
  \includegraphics{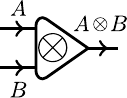} &
  \includegraphics{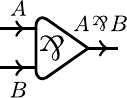} &
  \includegraphics{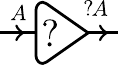}  \\
  contraction & cocontraction & weakening & coweakening \\
  \includegraphics{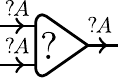} &
  \includegraphics{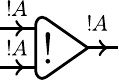} &
  \includegraphics{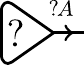} &
  \includegraphics{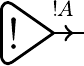} \\
\end{tabular}
% \caption{Cells}\label{figure:cells}
\end{minipage}
\begin{minipage}{.24\linewidth}
\centering
exponential box ($\oc_p$)
\includegraphics{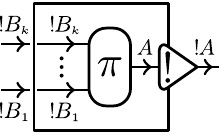}
% \caption{Exponential box}\label{figure:box}
\end{minipage}
\caption{Cells and box}\label{figure:nets}
\end{table}
\subsection{Reduction}

The reduction rules are illustrated in Table~\ref{figure:reduction}. Notice that
the exponential reductions are required to operate on closed boxes (boxes
without auxiliary doors). $\red{er}$ and $\red{e}$ may be performed  at any
depth. All the other rules are only allowed at the surface: reduction inside boxes
is prohibited. 

\proofset{straight,bottom=5pt}
% \begin{gather*}
% \begin{aligned}
\begin{table}[tb]
\centering
\begin{minipage}{0.54\textwidth}
\centering
  $\begin{gathered}
  %%%%%%%%%%%%%%%%%%%%%%%%%%%%
  % multiplicative redex
  %%%%%%%%%%%%%%%%%%%%%%%%%%%%
  \tikztarget(multexp){\extm{multexp}{
    \begin{proofnet}[baseline=-.5ex]
      \init[grow=right]{
        \binary[bottom=2.5pt]{\parr}{
          \wire{}
        }{
          \wire{}
        }
      }
      \init[grow=left]{
        \binary[bottom=2.5pt]{\otimes}{
          \wire{}
        }{
          \wire{}
        }
      }
      \end{proofnet}
    \red m
      %%%%%%%%%%%%%%%%%%%%%%%%%%%%%
      % multiplicative reduct
      %%%%%%%%%%%%%%%%%%%%%%%%%%%%%
      \begin{proofnet}[baseline=-.5ex,rotate=90]
      \draw (5pt,-10pt) -- ++(0,20pt);
      \draw (-5pt,-10pt) -- ++(0,20pt);
      \end{proofnet}
      \qquad\hspace{20pt}
      %%%%%%%%%%%%%%%%%%%%%%%%%%%%%
      % der vs box redex
      %%%%%%%%%%%%%%%%%%%%%%%%%%%%%
      \begin{proofnet}[baseline=-.5ex]
      \init[grow=left]{
        \unary[flip]{\wn}{
        \boxpal{p}{
          \nary[shape=rectangle,cell={net={20pt}{12pt}}]{\pi}{}
        }
        }
      }
      \drawbox{p}
      \end{proofnet}
    \red e
      %%%%%%%%%%%%%%%%%%%%%%%%%%%%%%
      %% der vs box reduct
      %%%%%%%%%%%%%%%%%%%%%%%%%%%%%%
      \begin{proofnet}[baseline=-.5ex]
      \init[grow=left]{
          \nary[shape=rectangle,cell={net={20pt}{12pt}}]{\pi}{}
      }
      \end{proofnet}
  }}
  \\[1ex]
  \tikzvcenter{\tikztarget(duperase){\extm{duperase}{
  \begin{proofnet}[baseline=-.5ex]
    \init[grow=right]{
      \binary{\wn}{\wire{}}{\wire{}}
    }
    \init[grow=left]{
      \boxpal[bottom=0pt]{p}{
        \nary[shape=rectangle,cell={net={17.5pt}{12pt}}]{\pi}{}
      }
    }
    \drawbox p
  \end{proofnet}
  \red {d} 
  \begin{proofnet}[baseline=-.5ex]
    \init[multi,grow=left,sep=8pt]{
      \boxpal[]{p}{
        \nary[shape=rectangle,cell={net={15pt}{12pt}}]{\pi}{}
      },
      \boxpal[]{q}{
        \nary[shape=rectangle,cell={net={15pt}{12pt}}]{\pi}{}
      }
    }
    \drawbox p
    \drawbox q
  \end{proofnet}
  \qquad\hspace{15pt}
  \begin{proofnet}[baseline=-.5ex]
    \init[grow=right]{
      \zeroary{\wn}
    }
    \init[grow=left]{
      \boxpal[bottom=0pt]{p}{
        \nary[shape=rectangle,cell={net={17.5pt}{12pt}}]{\pi}{}
      }
    }
    \drawbox p
    \end{proofnet}
  \red {er} 
  \epsilon
  }}}
  % \qquad
  % \tikzvcenter{\tikztarget(neutr){\extm{neutr}{
  % \begin{proofnet}[baseline=-.5ex]
  %   \init[grow=right]{
  %     \binary[very very small]{\wn}{\coord(c)}{\unary[very very small]{\wn}{\point(a)}}
  %   }
  %   \init(c){\wireupto(a){}}
  % \end{proofnet}
  % \red s
  % \begin{proofnet}[baseline=-.5ex,rotate=90]
  % \draw(0,0) -- ++(0,15pt);
  % \end{proofnet}
  % }}}
  \\[1ex]
  \tikzvcenter{\tikztarget(boxcompose){\extm{boxcompose}{
  %%%%%%%%%%%%%%%%%%%%%%%%%%%%%
  % box vs box redex
  %%%%%%%%%%%%%%%%%%%%%%%%%%%%%
    \begin{proofnet}[baseline=-.5ex]
    \init[grow=left]{
      \boxpal{p}{
        \nary[shape=rectangle,cell={net={37.5pt}{12pt}}]{\pi}{
            \wire[inbox={}]{
              \boxpal q{
                \nary[shape=rectangle,cell={net={17.5pt}{12pt}},bottom=3pt]{\sigma}{
                  {\wire[wire=fake,inbox={},bottom=7.5pt]{\point(a)}}
                }
              }
            },,
          {
            \wireupto[straight,inbox={}](a){\coord(d)}
          },
          {
            \wireupto[straight,inbox={}](a){\coord(c)}
          }
      }
    }
    }
    \drawbox{p}
    \drawbox{q}
    \pdots[2pt](c)(d)
    \end{proofnet}
  \red c
  %%%%%%%%%%%%%%%%%%%%%%%%%%%%%
  % box vs box reduct
  %%%%%%%%%%%%%%%%%%%%%%%%%%%%%
    \begin{proofnet}[baseline=-.5ex]
    \init[grow=left]{
      \boxpal{p}{
        \nary[shape=rectangle,cell={net={37.5pt}{12pt}}]{\pi}{
            {\wire[inbox={},bottom=0pt]{
              \boxpal[] q{
                \nary[shape=rectangle,cell={net={17.5pt}{12pt}},bottom=3pt]{\sigma}{
                  {\wire[wire=fake,inbox={},bottom=12.5pt]{\point(a)}},
                }
              }
            }}
            ,,
          {
            \wireupto[straight,inbox={}](a){\coord(d)}
          },
          {
            \wireupto[straight,inbox={}](a){\coord(c)}
          }
      }
    }
    }
    \drawbox[inbox=p]{q}
    \drawbox{p}
    \pdots[2pt](c)(d)
    \end{proofnet}
  }}}
  %%%%%%
  \end{gathered}$
\end{minipage}
\begin{minipage}{0.44\linewidth}
  $\begin{gathered}
  % \qquad
  % \tikzvcenter{\tikztarget(weakborder){\extm{weakborder}{
  % \begin{proofnet}[baseline=-.5ex]
  %   \init[grow=right]{
  %     \binary[very very small,bottom=15pt]{\wn}{\wire[wire=fake]{\point(a)}}{}
  %   }
  %   \draw[railed] (7.5pt,12.5pt) coordinate (x)-- ++(0,-17.5pt) coordinate (y);
  %   \draw[densely dotted] (y) -- ++(0,-7.5pt);
  %   \draw (x)  -- ([xshift=-7.5pt]x-|a) coordinate (y);
  %   \draw[densely dotted] (y) -- (y-|a);
  % \end{proofnet}
  % \red s
  % \begin{proofnet}[baseline=-.5ex]
  %   \init[grow=right]{
  %     \binary[very very small]{\wn}{\wire[bottom=20pt,wire=fake]{\point(a)}}{}
  %   }
  %     \draw[railed] (22.5pt,12.5pt) coordinate (x)-- ++(0,-17.5pt) coordinate (y);
  %     \draw[densely dotted] (y) -- ++(0,-7.5pt);
  %     \draw (x) -- ([xshift=-7.5pt]x-|a) coordinate (y);
  %     \draw[densely dotted] (y) -- (y-|a);
  %   \end{proofnet}
  % }}}
  \tikztarget(choice){\extm{choice}{
    \begin{proofnet}[baseline=-.5ex]
      \init[grow=right]{\unary[flip]{\wn}{\binary{\oc}{\wire{}}{\wire{}}}}
    \end{proofnet}
    \red{nd} 
    \begin{proofnet}[baseline=-.5ex]
      \coordinate (m) at (-30pt,0);
      \init[grow=left,multi]{
        \unary{\wn}{\wireto(m){}},
        \zeroary{\wn}
      }
    \end{proofnet}
    +
    \begin{proofnet}[baseline=-.5ex]
      \coordinate (m) at (-30pt,0);
      \init[grow=left,multi]{
        \zeroary{\wn},
        \unary{\wn}{\wireto(m){}}
      }
    \end{proofnet}
    % \qquad
    % \begin{proofnet}[baseline=-.5ex]
    %   \init[grow=right]{\unary[flip]{\wn}{\zeroary{\oc}}}
    % \end{proofnet}
    % \red 0
    % 0
  }}
  \\[1ex]
  \tikztarget(bialgebra){\extm{bialgebra}{
    \begin{proofnet}[baseline=-.5ex]
      \init[grow=left]{\binary{\oc}{\wire{}}{\wire{}}}
      \init[grow=right]{\binary[bottom=0pt]{\wn}{\wire{}}{\wire{}}}
    \end{proofnet}
    \red {ba}
    \begin{proofnet}[baseline=-.5ex]
      \init[multi,grow=left]{
        {\binary{\oc}{\coord(a)}{\coord(b)}},
        {\binary{\oc}{\coord(c)}{\coord(d)}}
      }
      \ljoin ac[bottom=10pt,wires={,railed}]{\wn}{\wire{}}
      \rjoin bd[bottom=10pt,wires={railed,}]{\wn}{\wire{}}
    \end{proofnet}
    % &
  }}
  \\[1ex]
  \tikztarget(srules){\extm{srules}{
  \begin{aligned}
    \begin{proofnet}[baseline=-.5ex]
      \init[grow=left]{\zeroary{\oc}}
      \init[grow=right]{\binary[bottom=0pt]{\wn}{\wire{}}{\wire{}}}
    \end{proofnet}
    &\red {s_1}
    \begin{proofnet}[baseline=-.5ex]
      \init[multi,grow=left]{
        {\zeroary{\oc}},
        {\zeroary{\oc}}
      }
    \end{proofnet}
    &\hspace{20pt}
    \begin{proofnet}[baseline=-.5ex]
      \init[grow=left]{\binary{\oc}{\wire{}}{\wire{}}}
      \init[grow=right]{\zeroary[bottom=0pt]{\wn}}
    \end{proofnet}
    &\red {s_2}
    \begin{proofnet}[baseline=-.5ex]
      \init[multi,grow=right]{
        {\zeroary{\wn}},
        {\zeroary{\wn}}
      }
    \end{proofnet}
    \\
    \begin{proofnet}[baseline=-.5ex]
      \init[grow=left]{\zeroary{\oc}}
      \init[grow=right]{\zeroary[bottom=0pt]{\wn}}
    \end{proofnet}
    &\red \epsilon
    \epsilon
    \end{aligned}
  }}
  %%%%%%
  \end{gathered}$
\end{minipage}

\caption{Reduction rules}\label{figure:reduction}
\end{table}

%%%% OLD PART ON NON-SQUIGGLED ARROW AND STUFF %%%%% BEGIN 

% We define the structural reduction $\red s$ as the union of all rules, excepted
% the logical rules $\red m$ and $\red e$. The box reduction comprises the rules
% handling boxes : $\red{\text{box}} := \red e \cup \red d \cup \red{er} \cup \red
% c$.  Dually, $\red{\text{nbox}}$ is composed of all the rules except thes ones.
% Finally, we define the differential reduction as the union of all the
% aformentionned rules $\red {\text{diff}} := \red e \cup \red m \cup \red s$.
%
% % \newcommand*{\rred}[1]{\mathrel{\sur{\leftarrow}{\mathtt{#1}}}} % reversed red
% % many-steps reverse reduction
% % \newcommand*{\rredstar}[1]{\mathrel{^{\ast}{\sur{\leftarrow}{\mathtt{#1}}}}}
%
% The surface and closed reduction have a very strong commutation property as it
% does not contain any critical pair:
%
% \begin{property}{Diamond}\label{property-nets-diamond}\\
%   The rewriting relation $\redsurf$ has the diamond property: If
%   $\net{R}'
%   \rred{\text{diff}} \net{R} \redsurf \net{R}''$, then there exists
%   $\net S$ such that $\net{R}' \redsurf \net{S} \rred{\text{diff}}
%   \net{R}''$
% \end{property}
%
% In the context of non correct nets, where termination is not assured, the diamond
% property implies the following important fact:
%
% \begin{mcor}{Termination}\label{mcor-nets-termination}\\
%   A net is weakly normalizing if and only if it is strongly normalizing
% \end{mcor}
%

%%%% OLD PART ON NON-SQUIGGLED ARROW AND STUFF %%%%% END 

%%%% OLD SUBSECTION : rednsurf REDUCTION %%%%
% \subsection{Equivalence relation and $\rednsurf$ reduction}
We quotient the proof nets by associativity and commutativity of contraction and
cocontractions. We also include in the relation the neutrality of (co)weakening for
(co)contraction (Table~\ref{figure:equivalence}).
\proofset{straight,bottom=5pt}
% \begin{gather*}
% \begin{aligned}
\begin{table}[tb]
\centering
$\begin{gathered}
\\[2ex]
\tikztarget(equivwn){\extm{equivwn}{
\begin{proofnet}[baseline=-.5ex]
  \init[grow=right]{
    \binary{\wn}{\wire{}}{\wire{}}
  }
\end{proofnet}
  \equiv
\begin{proofnet}[baseline=-.5ex]
  \init[grow=right]{
    \binary{\wn}{\wire[bottom=10pt,straight=false,shift=11pt]{}}{\wire[bottom=10pt,straight=false,shift=-11pt,wire=railed]{}}
  }
\end{proofnet}
\quad\quad %SEP
\begin{proofnet}[baseline=-.5ex]
  \init[grow=right]{
    \binary{\wn}{\binary{\wn}{\wire{\point(a)}}{\wire{}}}{\coord(c)}
  }
  \init(c){\wireupto(a){}}
\end{proofnet}
  \equiv
\begin{proofnet}[baseline=-.5ex]
  \init[grow=right]{
    \binary{\wn}{\coord(c)}{\binary{\wn}{\wire{\point(a)}}{\wire{}}}
  }
  \init(c){\wireupto(a){}}
\end{proofnet}
\quad\quad %SEP
  \begin{proofnet}[baseline=-.5ex]
    \init[grow=right]{
      \binary{\wn}{\zeroary[name=weak]{\wn}}{\wireupto(weak.middle pax){}}
    }
  \end{proofnet}
\equiv
  \begin{proofnet}[baseline=-.5ex,rotate=90]
    \draw (0pt,-10pt) -- ++(0,20pt);
  \end{proofnet}
}}
\\[2ex]
\tikztarget(equivoc){\extm{equivoc}{
\begin{proofnet}[baseline=-.5ex]
  \init[grow=right]{
    \binary{\oc}{\wire{}}{\wire{}}
  }
\end{proofnet}
  \equiv
\begin{proofnet}[baseline=-.5ex]
  \init[grow=right]{
    \binary{\oc}{\wire[bottom=10pt,straight=false,shift=11pt]{}}{\wire[bottom=10pt,straight=false,shift=-11pt,wire=railed]{}}
  }
\end{proofnet}
\quad\quad%SEP
\begin{proofnet}[baseline=-.5ex]
  \init[grow=right]{
    \binary{\oc}{\binary{\oc}{\wire{\point(a)}}{\wire{}}}{\coord(c)}
  }
  \init(c){\wireupto(a){}}
\end{proofnet}
  \equiv
\begin{proofnet}[baseline=-.5ex]
  \init[grow=right]{
    \binary{\oc}{\coord(c)}{\binary{\oc}{\wire{\point(a)}}{\wire{}}}
  }
  \init(c){\wireupto(a){}}
\end{proofnet}
\quad\quad %SEP
  \begin{proofnet}[baseline=-.5ex]
    \init[grow=right]{
      \binary{\oc}{\zeroary[name=weak]{\oc}}{\wireupto(weak.middle pax){}}
    }
  \end{proofnet}
\equiv
  \begin{proofnet}[baseline=-.5ex,rotate=90]
    \draw (0pt,-10pt) -- ++(0,20pt);
  \end{proofnet}
}}
\end{gathered}$
\caption{Equivalence relation}\label{figure:equivalence}
\end{table}

We denote by $\to$ the full reduction, extended to sums of nets in the same way
as in~\cite{ehrhard:hal-00150274}.

\begin{theorem}{Confluence}\label{theorem-nets-confluence}\\
  The reduction $\to$ is confluent.
\end{theorem}

%%%%%%%%%%%%%%%%%%%% --- EXTENDED --- %%%%%%%%%%%%%%%%%%%
\begin{proof}
  The reduction is similar in spirit to an orthogonal reduction in a term
  rewriting system.  In the same approach as the proof of confluence of the
  $\lambda$-caluclus, we can define a parallel reduction $\redpar$, which allow
  to perform an abritrary number of step in parallel (meaning that we can't
  reduce redexes that are created by other the other steps). This parallel
  reduction verifies $\to \subseteq
  \redpar \subseteq \to^\ast$. Then we prove that $\redpar$ has the diamond
  property, and the previous expression ensures that $\to^\ast \subseteq
  \redpar^\ast$, thus $\to^\ast$ is confluent.
\end{proof}

\subsection{Termination}
%%%%%%%%%%%%%%%%%% --- END EXTENDED --- %%%%%%%%%%%%%%%%%

\begin{theorem}\label{theorem-nets-termination}
  A net is weakly normalizing if and only if it is strongly
  normalizing.
\end{theorem}

%%%%%%%%%%%%%%%%%%%% --- EXTENDED --- %%%%%%%%%%%%%%%%%%%

\newcommand{\redsurf}{\rightharpoonup}
\newcommand{\redsurfinv}{\leftharpoonup}
\newcommand{\rednsurf}{\leadsto}
\newcommand{\rednsurfinv}{\leftharpoondown}
\tikzset{
  equiv/.style={
    draw=none,
    edge node={node [sloped, allow upside down, auto=false]{$\rednsurf^\ast$}}},
  equiv1step/.style={
    draw=none,
    edge node={node [sloped, allow upside down, auto=false]{$\rednsurf$}}},
  equiv01step/.style={
    draw=none,
    edge node={node [sloped, allow upside down,
      auto=false]{$\underline{\rednsurf}$}}},
  equivkstep/.style={
    draw=none,
    edge node={node [sloped, allow upside down,
      auto=false]{$\rednsurf^k$}}},
  equivkpstep/.style={
    draw=none,
    edge node={node [sloped, allow upside down,
    auto=false]{$\rednsurf^{k'}$}}},
  equals/.style={
    draw=none,
    edge node={node [sloped, allow upside down, auto=false]{$=$}}},
  reduceto/.style={
    draw=none,
    edge node={node [sloped, allow upside down, auto=false]{$\redsurf^\ast$}}},
  reduceto1step/.style={
    draw=none,
    edge node={node [sloped, allow upside down, auto=false]{$\redsurf$}}},
  reduceto1pstep/.style={
    draw=none,
    edge node={node [sloped, allow upside down, auto=false]{$\redsurf^+$}}},
  reducetokstep/.style={
    draw=none,
    edge node={node [sloped, allow upside down, auto=false]{$\redsurf^k$}}},
  reducetokpstep/.style={
    draw=none,
    edge node={node [sloped, allow upside down, auto=false]{$\redsurf^{k'}$}}}
}

To prove Theorem~\ref{theorem-nets-termination}, we proceed in two steps. First, we
consider the system where all reductions are restricted to the surface.
This reduction is constrained enough so that it verifies the diamond
property, which is known to implies the 
Theorem~\ref{theorem-nets-termination}. Then, we consider the non-surface
$\red{er}$ and $\red{e}$ reductions, that we write $\rednsurf$, and show that
adding it to $\redsurf$ does not invalid the theorem. In this subsection, we denote by
$\redsurf$ all the reductions rules of Figure~\ref{figure:reduction} applied
at the surface. We have the decomposition $\to$ = $\redsurf \cup \rednsurf$.
The following lemma shows that $\redsurf$ and $\rednsurf$ have a 
commutation property:

\begin{lemma}{Swapping}\label{lemma-nets-swapping}\\
  Let $\mathcal{S} \rednsurf^\ast \mathcal{R}$ and $\mathcal{R} \redsurf^+ \mathcal{R}'$. Then
  there exists $\mathcal{S'}$ such that

  \begin{center}
    \begin{tikzpicture}
      \matrix (m) [matrix of math nodes,row sep=2em,column sep=2em,minimum
        width=2em] 
      {
        \mathcal{R} & \mathcal{R}' \\
        \mathcal{S} & \mathcal{S'} \\
      };
      \path[-stealth, auto]
        (m-1-1) edge[reduceto1pstep] (m-1-2)
        (m-2-1) edge[reduceto1pstep] (m-2-2)
        (m-2-1) edge[equiv] (m-1-1)
        (m-2-2) edge[equiv] (m-1-2);
    \end{tikzpicture} 
  \end{center}
\end{lemma}

\begin{proof}Lemma~\ref{lemma-nets-swapping}\\
  By induction on both the length of the reduction $\mathcal{S} \rednsurf^\ast
  \mathcal{R}$ and $\mathcal{R} \to^+ \mathcal{R}'$. We consider the
  fundamental case where both reductions are one-step : 
  
  \begin{center}
    \begin{tikzpicture}
      \matrix (m) [matrix of math nodes,row sep=2em,column sep=2em,minimum
        width=2em] 
      {
        \mathcal{R} & \mathcal{R}' \\
        \mathcal{S} & \\
      };
      \path[-stealth, auto]
        (m-1-1) edge[reduceto1step] (m-1-2)
        (m-2-1) edge[equiv1step] (m-1-1);
    \end{tikzpicture} 
  \end{center}

  We make the following remarks:
  \begin{itemize}
    \item As $\rednsurf$ reduction happens inside boxes, and $\redsurf$ outside,
      we can put in a one-to-one correspondence the boxes at depth 0 - and thus
      the $\redsurf$-redexes - of $\mathcal{R}$ and $\mathcal{S}$ 
    \item If the reduction in $\mathcal{R}$ does not involve a box, it does not interact with
      $\rednsurf$, and the diagram can be closed as a square with one step reduction in every
      direction
    \item If the reduction in $\net R$ erases a box, we can erase the
      corresponding one in $\mathcal{S}$, which may delete the $\rednsurf$ redex
      involved reduced in $\net S$.  In any case we can reduce $\mathcal{S}$ to
      $\mathcal{S}'$ in one step by erasing this box and $\mathcal{S}'$ to
      $\mathcal{R}'$ by at most one $\rednsurf$ step.
    \item If the reduction in $\net R$ opens a box, we can open the
      corresponding one in $\net S$. If the $\rednsurf$-redex was at depth $1$
      in the same box, then it becomes a $\redsurf$-redex as it now appears at
      the surface. Otherwise, it can be performed in one $\rednsurf$ step. In
      any case, we can close the diagram by peforming either one $\redsurf$ step
      followed by one $\rednsurf$ step, or two $\redsurf$ steps.
  \end{itemize}

  In the cases listed above, the length of the reduction between
  $\mathcal{R}'$ and $\mathcal{S}'$ never exceeds one, hence we can always
  fill the following diagram :
    
    \begin{center}
      \begin{tikzpicture}
        \matrix (m) [matrix of math nodes,row sep=2em,column sep=2em,minimum
          width=2em] 
        {
          \mathcal{R} & \mathcal{R}' \\
          \mathcal{S} & \mathcal{S'} \\
        };
        \path[-stealth, auto]
          (m-1-1) edge[reduceto1step] (m-1-2)
          (m-2-1) edge[reduceto1pstep] (m-2-2)
          (m-2-1) edge[equiv1step] (m-1-1)
          (m-2-2) edge[equiv01step] (m-1-2);
      \end{tikzpicture} 
    \end{center}
  
  Where the bottom line does not involve duplication
  ($\mathcal{S}' \underline{\rednsurf} \mathcal{R}'$ if $\mathcal{S}' =
  \mathcal{R}'$ or $\mathcal{S'} \rednsurf \mathcal{R}'$).

  \begin{itemize}
    \item Duplication is the only $\redsurf$ reduction that can create new
      $\rednsurf$ redexes, but the bottom $\redsurf$ reduction requires exactly
      one step in $\mathcal{S}$. This corresponds to the
      following diagram: 
      
      \begin{center}
        \begin{tikzpicture}
          \matrix (m) [matrix of math nodes,row sep=2em,column sep=2em,minimum
            width=2em] 
          {
            \mathcal{R} & \mathcal{R}' \\
            \mathcal{S} & \mathcal{S'} \\
          };
          \path[-stealth, auto]
            (m-1-1) edge[reduceto1step] (m-1-2)
            (m-2-1) edge[reduceto1step] (m-2-2)
            (m-2-1) edge[equiv1step] (m-1-1)
            (m-2-2) edge[equiv] (m-1-2);
        \end{tikzpicture} 
      \end{center}
  \end{itemize}

  Let us now prove the lemma. We consider the two cases separately : let us assume that $\mathcal{R}
  \redsurf \mathcal{R'}$ is not a duplication.  We perform an induction
  on the length of the reduction $\mathcal{S} \rednsurf^k \mathcal{R}$. The
  induction hypothesis is that we can always fill the following diagram

  \begin{center}
      \begin{tikzpicture}
        \matrix (m) [matrix of math nodes,row sep=2em,column sep=2em,minimum
          width=2em] 
        {
          \mathcal{R} & \mathcal{R}' \\
          \mathcal{S} & \mathcal{S'} \\
        };
        \path[-stealth, auto]
          (m-1-1) edge[reduceto1step] (m-1-2)
          (m-2-1) edge[reduceto] (m-2-2)
          (m-2-1) edge[equivkstep] (m-1-1)
          (m-2-2) edge[equivkpstep] (m-1-2);
      \end{tikzpicture} 
  \end{center}

  With $k' \leq k$, and where the bottom reduction does not contain any
  duplication step.  The base case $k=0$ is trivially true. For the induction
  step, if we have $\mathcal{S} \rednsurf^k \mathcal{T} \rednsurf \mathcal{R}$,
  we use the first diagram of the remarks to get the middle line $T \redsurf^\ast
  T_p$.  Then, we apply the induction hypothesis on each reduction step $T_i \to
  T_{i+1}$, which we can do precisely because our IH states that the lengths
  $k_i$ of each reduction $S_i \rednsurf^{k_i} T_i$ verify $k_i \leq k_{i-1}
  \leq \ldots \leq k$. We paste all the diagrams and get 

  \begin{center}
      \begin{tikzpicture}
        \matrix (m) [matrix of math nodes,row sep=3em,column sep=1em,minimum
          width=1em] 
        {
          \mathcal{R} & & & \mathcal{R}' \\
          \mathcal{T} & \mathcal{T}_1 & \phantom{|}\ldots & \mathcal{T}_p \\
          \mathcal{S} & \mathcal{S}_1 & \phantom{|}\ldots & \mathcal{S}_p \\
        };
        \path[-stealth, auto]
          (m-1-1) edge[reduceto1step] (m-1-4)
          (m-2-1) edge[reduceto1step] (m-2-2)
          (m-2-2) edge[reduceto1step] (m-2-3)
          (m-2-3) edge[reduceto1step] (m-2-4)
          (m-3-1) edge[reduceto1step] (m-3-2)
          (m-3-2) edge[reduceto1step] (m-3-3)
          (m-3-3) edge[reduceto1step] (m-3-4)
          (m-2-1) edge[equiv1step] (m-1-1)
          (m-2-4) edge[equiv01step] (m-1-4)
          (m-3-1) edge[equiv] (m-2-1)
          (m-3-2) edge[equiv] (m-2-2)
          (m-3-4) edge[equiv] (m-2-4);
      \end{tikzpicture} 
  \end{center}

  The case of duplication is simpler as the step $\mathcal{R} \redsurf \mathcal{R}'$
  is reflected by just one step $\mathcal{S} \redsurf \mathcal{S}'$ in the second
  diagram of the remarks. Our IH is now that we can fill the diagram, without
  further assertion on the length of the $\rednsurf$ reduction. Indeed, if $\mathcal{S} \rednsurf^k \mathcal{T}
  \rednsurf \mathcal{R}$, we use the one-step diagram to get $\mathcal{T}'$ and
  apply the IH to fill the bottom part with $\mathcal{S}'$ : 

  \begin{center}
      \begin{tikzpicture}
        \matrix (m) [matrix of math nodes,row sep=2em,column sep=2em,minimum
          width=2em] 
        {
          \mathcal{R} & \mathcal{R}' \\
          \mathcal{T} & \mathcal{T}' \\
          \mathcal{S} & \mathcal{S}' \\
        };
        \path[-stealth, auto]
          (m-1-1) edge[reduceto1step] (m-1-2)
          (m-2-1) edge[reduceto1step] (m-2-2)
          (m-3-1) edge[reduceto1step] (m-3-2)
          (m-2-1) edge[equiv1step] (m-1-1)
          (m-3-1) edge[equivkstep] (m-2-1)
          (m-2-2) edge[equiv] (m-1-2)
          (m-3-2) edge[equiv] (m-2-2);
      \end{tikzpicture} 
  \end{center}

  Finally, we can perform a second induction on the length of the reduction
  $\mathcal{R} \to^+ \mathcal{R}'$ to get the final result.
\end{proof}

Writing a reduction $\net R \to^\ast S$ as blocks $\net R \rednsurf^\ast R_1
\redsurf^\ast R_2 \rednsurf^\ast \ldots \redsurf^\ast R_n$,
we can iterate Lemma~\ref{lemma-nets-swapping} to form a new reduction sequence with
only two distinct blocks:

\begin{lemma}{Standardization}\label{lemma-nets-standardization}\\
  Let $\net R \to^\ast \net S$. Then $\net R \redsurf^\ast \net{R}'
  \rednsurf^\ast \net S$. Moreover, if the original reduction contains at least
  one $\redsurf$ step, then $\net{R} \redsurf^+ \net{R}'$.
\end{lemma}

\begin{proof}Lemma~\ref{lemma-nets-standardization}\\
This follows from the previous lemma: we decompose the reduction $\net R
\to^\ast S$ as blocks $\net R \rednsurf^\ast R_1 \redsurf^\ast R_2
\rednsurf^\ast \ldots \redsurf^\ast R_n$, and iterate
\ref{lemma-nets-swapping} to gather the reductions into only two distinct
blocks.
\end{proof}

Lemma~\ref{lemma-nets-rednsurf-neutrality} follows from
the observation that as $\redsurf$ only acts on surface and $\rednsurf$
inside boxes, the latter can not interact with the redexes of the former.

\begin{lemma}{Neutrality of $\rednsurf$}\label{lemma-nets-rednsurf-neutrality}\\
  $\rednsurf$ does not create nor erase $\redsurf$-redexes. In
  particular, if $\net R \rednsurf^\ast \net{R}'$, then $\net{R}$ is
  $\redsurf$-normal if and only if $\net{R}'$ is.
\end{lemma}

We also need some properties about termination of the reduction $\redsurf$ and
$\rednsurf$ :

\begin{lemma}{Strong normalization for $\rednsurf$}\label{lemma-nets-termination-nsurf}\\
  $\rednsurf$ is strongly normalizing.
\end{lemma}

\begin{proof}
  It suffices to note that opening or deleting a box strictly decreases the
  total number of boxes in the net.
\end{proof}

\begin{lemma}\label{lemma-nets-termination-surf}
  Theorem~\ref{theorem-nets-termination} is true when replacing $\to$ with
$\redsurf$.
\end{lemma}

\begin{proof}
  The surface reduction satisfies the diamond property,
  which excludes the existence of a weakly normalizing term with an
  infinite reduction.
\end{proof}

We can finally prove Theorem~\ref{theorem-nets-termination}:

\begin{proof}Theorem~\ref{theorem-nets-termination}\\
  We prove two auxiliary properties:
\begin{description}
  \item[(a)] \textbf{$\to$-weak normalization implies $\redsurf$-weak
    normalization} \\
    Let $\net R$ be $\to$-weakly normalizing, $\net R \to^\ast \net N$ a
    reduction to its normal form. By \ref{lemma-nets-standardization}, we can write
    $\net R \redsurf^\ast \net S \rednsurf^\ast N$. $N$ being a
    $\to$-normal form, it is also a $\redsurf$-normal form, and by
    \ref{lemma-nets-rednsurf-neutrality} so is $\net S$. $\net R$ is thus a
    $\redsurf$-weakly normalizing.
  \item[(b)] \textbf{an infinite $\to$-reduction gives an infinite
    $\redsurf$-reduction}\\
    Let $\net R$ be a net with an infinite reduction, written $\net R \to^\ast
    \infty$. We will prove by induction that for any $n \geq 0$, there exists
    $\net{R}_n$ such that $R \redsurf^n \net{R}_n \to^\ast \infty$.
    \begin{itemize}
      \item Case $n = 0$\\
        We just take $\net{R}_0 = \net R$
      \item Inductive case\\ If $\net R \redsurf^n \net{R}_n \to^\ast \infty$,
        we take any infinite reduction starting from $\net{R}_n$. If the first
        step is $\net{R}_n \redsurf \net S$, then we take $\net{R}_{n+1} = \net
        S$. Otherwise, the first step is a $\rednsurf$ step, and we take the
        maximal block of $\rednsurf$ reduction starting from $\net{R}_n$. By
        \ref{lemma-nets-termination-nsurf}, this block must indeed be finite and we can write $\net{R}_n
        \rednsurf^\ast \net S \redsurf \net{S}' \to \infty$. By \ref{lemma-nets-swapping},
        we can swap the two blocks such that $\net{R}_n \redsurf \net{R}'
        \redsurf^\ast
        \net{R}'' \rednsurf^\ast \net{S}'$, and we take $\net{R}_{n+1} =
        \net{R}'$.
    \end{itemize}
\end{description}
From these two points follows that $\to$-weak normalization implies $\to$-strong
normalization. If a net is $\to$-weakly normalizing, then it is
$\redsurf$-weakly normalizing by (a). By
\ref{lemma-nets-termination-surf}, it is also $\redsurf$-strongly normalizing.
But by (b) it must be also $\to$-strongly normalizing.
\end{proof}
%%%%%%%%%%%%%%%%%% --- END EXTENDED --- %%%%%%%%%%%%%%%%%

%
% Summing up: 
% \begin{itemize}
%   \item A $\to$-weakly normalizing net is also $\redsurf$-weakly
%     normalizing.
%   \item Using the termination property for $\redsurf$, a $\redsurf$-weakly normalizing net is
%     $\redsurf$-strongly normalizing.
%   \item On the other hand, an infinite $\to$-reduction of a net $\net R$ can be
%     transformed to an infinite $\redsurf$-reduction using
%     \ref{lemma-nets-standardization}. Thus, a $\redsurf$-strongly normalizing
%     term must actually be $\to$-strongly normalizing.
% \end{itemize}
%
% From which we can conclude:
%
% \begin{theorem}{Termination for $\to$}\label{theorem-nets-termination}\\
%   A net is $\to$-weakly normalizing if and only if it is $\to$-strongly
%   normalizing.
% \end{theorem}
%
The next section focuses on a specific family of proof nets that play a key role
in the expression of communication primitives inside proof nets.

\section{Routing Areas}
\label{section:areas}

Let us now define and study a special kind of nets: the
\emph{routing areas}. It is a generalization of the construction of
communication areas introduced in \cite{EHRHARD2010606}. The approach is
similar: we aim at constructing nets to be used as building blocks to implement
communication primitives. We shall see that this seemingly restricted class of nets is
actually the set of normal forms of a fragment of proof nets
(Theorem~\ref{theorem-mrarea-charact}). Routing areas are composed only of structural
rules: contraction, weakening, cocontraction and coweakening. These basic
components act as resource dispatchers (a resource designates a closed exponential box in
the following):

\begin{itemize}
  \item
    $\tikzvcenter{\begin{proofnet}\init[grow=left]{\wire{}}\end{proofnet}}$\\ A
    free wire acts as the identity. It passively forwards a resource
    that is connected on the \emph{input} (the left port) to the \emph{output}
    (the right port).
  \item 
    $\tikzvcenter{\begin{proofnet}
    \init[grow=right]{\binary{\wn}{\wire{}}{\wire{}}} \end{proofnet}}$ \\ A
    contraction is a broadcaster with one input and two outputs. A resource
    connected on the left will be copied to both outputs on the right. A
    weakening is a degenerate case of a broadcaster with zero outputs as
    broadcasting something to no one is the same as erasing it.
    % Stacking up contractions, a contraction tree with $n$ leafs is a
    % broadcaster that copies anything coming from its one input to all of the
    % $n$ outputs.
  \item 
    $\tikzvcenter{\begin{proofnet}
      \init[grow=left]{\binary{\oc}{\wire{}}{\wire{}}}
    \end{proofnet}}$ \\
    Dually, a cocontraction is a packer with two inputs and one output. A packer
    aggregates its two inputs non deterministically. When a dereliction is
    connected to the output to consume two packed resources, a non-deterministic
    sum of the two possible choices for the one to be provided is produced.
    Similarly, coweakening is seen as a degenerate packer with no inputs.
\end{itemize}

% A question which naturally arises is does this intuitive interpretation as a
% wiring between inputs and outputs scale to an arbitrary combination of these
% cells, that is a routing net ? Theorem~\ref{theorem-mrarea-charact} gives a
% positive answer, with reasonable requirements in the definition of a routing net
% (Definition~\ref{definition-mrarea-rnets}).
%
A routing area can be seen as a router, or a circuit, between inputs and
outputs.  Inputs are connected to contractions which broadcast resources they
receive to cocontractions. Cocontractions may gather resources from
multiple such sources.  The conclusion of these cocontractions form the outputs.
A routing area is then described by a slight generalization of a relation
between sets, a
\emph{multirelation}.  Its role is to define the \emph{wiring diagram} which
specifies which inputs and outputs are connected. Let us first introduce
multirelations:

% \newcommand{\labin}{\mathcal{L}_i}
% \newcommand{\labout}{\mathcal{L}_o}
% \newcommand{\cardn}[1]{| {#1} |}
% \newcommand{\routarea}[1]{\mathbf{#1}}

% Our device  then determined by the data of which input is connected to which output,
% with multiplicity, as an input could feed several copies to the same
% output. We encode this information through a relation with multiplicity, that we
% call a \emph{multirelation} :

\begin{description}
  \item[Multirelation] Let $A$ and $B$ be two sets, a \emph{multirelation} $R$
    between $A$ and $B$ is a multiset of elements of $A \times B$, or concretely
    a map $R : A \times B \to \mathbb{N}$.  For $k \in \mathbb{N}$, we write $x
    \mathrel{R_k} y$ if $R(x,y) = k$.
  \item[Relations and multirelations] 
    A relation $R$ between $A$ and $B$ can be seen as a multirelation by taking
    its characteristic function $\mathbbm{1}_{R}$.Conversely, we can forget the
    multiplicity of a multirelation $S$ and recover a relation by taking the
    subset of $A \times B$ defined by $\{ (x,y) \in A \times B \mid S(x,y) \geq
    1 \}$.
  \item[Composition]
    Multirelations enjoy a composition operation that computes all the ways
    to go from an element to another with multiplicities. For multirelations $R,S$
    respectively between $A$ and $B$, and $B$ and $C$,
    \[
      (S \circ R)(x,z) = \sum_{y \in B} R(x,y)S(y,z)
    \]
    This composition is associative, coincides with the usual one for relations,
    and has the identity relation (seen as a multirelation) for neutral. This is
    in fact the matrix multiplication, seeing a multirelation between $A$ and
    $B$ as a $\cardn{B} \times \cardn{A}$ matrix with integer coefficients
    $R(i,j)$ (identifying finite sets with their cardinal).
  \item[The \textbf{FMRel} Category] Finite sets and multirelations between them
    form a category \cat{FMRel}, with the category \cat{FRel} of finite sets and
    relations as a subcategory.  \cat{FMRel} has finite coproducts, extending
    the one of \cat{FRel}, and corresponding to the direct sum of matrices.
\end{description}

% We pictured here the two multirelations $R$ between $\{1,2,3\}$ and $\{1,2\}$, $S$
% between $\{1,2\}$ and $\{1,2,3\}$. We represents them as weighted graphs, with
% the domain on the left the codomain on the right. We connect a vertice $x$ of
% the domain to a vertice $y$ of the codomain with an edge weighted by $k =
% R(x,y)$ when $k > 0$.
%
% [SCHEMA]
%
% The composition is obtained by gluing the codomain of $R$ with the domain of
% $S$, and computing the sum $s = \sum_{p \in \text{Path}(x,y)} w(p)$ of the weights of paths between a vertice $x$ in the
% domain of $R$ and $y$ in the codomain of $S$ where the weight of a path is
% defined to be the product of the weights of the edges composing it.
% We add an edge between $x$ and $y$ weighted by $s$ if $s > 0$.

% \begin{proof}\ref{property-mrarea-coprod}\\
%   From $R : A \to C$ and $S : B \to C$, we can
%   form a multirelation $T : A + B \to C$
%   \[
%     T(x,c) =
%     \begin{cases}
%       R(a,c) & \text{ if } x = \textbf{inl}(a) \\
%       S(b,c) & \text{ if } x = \textbf{inr}(b)
%     \end{cases}
%   \]
%   One can easily check that this forms a coproduct in \cat{FMRel}, noted $T = R + S$.
% \end{proof}
%

A routing area is described by a multirelation between its inputs and its
outputs. Its value at the pair $(i,o)$ indicates how many times the input $i$ is
connected to the output $o$.  We are now ready to construct the \emph{routing
area} defined by a multirelation.

\begin{definition}{Routing area}\label{definition-routingarea}\\
  Let $\labin, \labout$ be two finite sets called the input labels and the
  output labels, and a multirelation $R$ between $\labin$ and $\labout$. A  
  routing area $\routarea R$ associated to the triplet $(\labin,\labout,R)$
  is a net constructed as follows:
  \begin{itemize}
    \item It has $\cardn{\labin} + \cardn{\labout}$ free wires partitioned into
      $\cardn{\labin}$ \emph{inputs} and $\cardn{\labout}$ \emph{outputs}. Each input is
      labelled by a distinct element of $\labin$, while outputs are labelled by
      distinct elements of $\labout$.
  % \begin{center}
  %   \tikzvcenter{
  %     \begin{proofnet}
  %       \node[net={40pt}{40pt}] (box) {\routarea R};
  %       \coordinate (i1) at ([yshift=-10pt]box.west);
  %       \coordinate (im) at ([yshift=10pt]box.west);
  %       \coordinate (o1) at ([yshift=-10pt]box.east);
  %       \coordinate (on) at ([yshift=10pt]box.east);
  %       \init{ 
  %         \wirefromto{[xshift=-10pt]i1}{[xshift=-10pt]i1}{i1}{i1}
  %         \wirefromto{[xshift=-10pt]im}{[xshift=-10pt]im}{im}{im}
  %         \wirefromto{[xshift=10pt]o1}{[xshift=10pt]o1}{o1}{o1}
  %         \wirefromto{[xshift=10pt]on}{[xshift=10pt]on}{on}{on}
  %       }
  %       \pdots[2pt]([xshift=-5pt]i1)([xshift=-5pt]im)
  %       \pdots[2pt]([xshift=5pt]o1)([xshift=5pt]on)
  %     \end{proofnet}
  %   }
  % \end{center}
    \item Each input (resp. output) is connected to the main port of a
      contraction (resp. cocontraction) tree. Then, for every $(i,o) \in \labin
      \times \labout$, we connect the tree of the input $i$ to the tree of
      output $o$ with exactly $R(i,o)$ wires.
  \end{itemize}
\end{definition}
\begin{center}
  \tikzvcenter{
    \begin{proofnet}
      % Tree of i1
      \init[grow=right]{
        \nary[name=i1]{\wn}{
          \wire{\point(i11)},
          \point(i12),
          \wire{\point(i13)}
        }
      }
      % Tree of im
      \init[grow=right]([yshift=7.5ex,xshift=-5pt]i1.pal){
        \nary[name=im]{\wn}{
          \wire{\point(im1)},
          \point(im2),
          \wire{\point(im3)}
        }
      }
      % Tree of o1
      \init[grow=left]([xshift=15ex]i1.pal){
        \nary[name=o1]{\oc}{
          \wire{\point(o11)},
          \point(o12),
          \wire{\point(o13)}
        }
      }
      % Tree of on
      \init[grow=left]([yshift=7.5ex,xshift=5pt]o1.pal){
        \nary[name=on]{\oc}{
          \wire{\point(on1)},
          \point(on2),
          \wire{\point(on3)}
        }
      }
      \pdots[2pt](i11)(i13)
      \pdots[2pt](im1)(im3)
      \pdots[5pt](i1.pal)(im.pal)
      \pdots[2pt](o11)(o13)
      \pdots[2pt](on1)(on3)
      \pdots[5pt](o1.pal)(on.pal)
    \end{proofnet}  
  }
\end{center}
We represent them as rectangular boxes, with the inputs appearing on the left and
outputs on the right.

%%%%%%%%%%%%%%%%%%%% --- EXTENDED --- %%%%%%%%%%%%%%%%%%%
\begin{definition}{Arity}\label{definition-routing-arity}\\
Let $(\labin, \labout, R)$ be a routing area. For an input $i \in \labin$, we
define its \emph{arity} as the number of leafs of the associated contraction
tree given by $\arity{i} = \sum_{o \in \labout} R(i,o)$. Similarly, the arity of
an output $o \in \labout$ is defined by $\arity{o} = \sum_{i \in \labin}
R(i,o)$. The set of outputs (resp. inputs) connected to an input $i$ (resp.
output $o$) is defined as $\conn{i} = \{ o \in \labout \mid R(i,o) \geq 1 \}$
(resp. $\conn{o} = \{ i \in \labin \mid R(i,o) \geq 1 \}$). In general, for an
input or an output $x$, $\arity{x} \leq \cardn{\conn{x}}$.  This is an equality
for all $x$ if and only if $R$ is a relation.
\end{definition}
%%%%%%%%%%%%%%%%%% --- END EXTENDED --- %%%%%%%%%%%%%%%%%

% Routing areas are defined up to relabelling : if $\routarea R$ is defined
% by $(\labin,\labout,R)$ then for any pair of permutation $\sigma : \labin' \to
% \labin$ and $\tau : \labout' \to \labout$, the conjugate relation
% $(\labin',\labout',\sigma R \tau^{-1})$ defines the same routing area,
% where $\sigma R \tau^{-1}$ is the composition of multirelations. Conversely, two
% descriptions $(\labin,\labout,R)$ and $(\labin',\labout',S)$ defining the same
% area must be conjugates: for each input wire, $\sigma$ sends the label
% attributed in $\labin'$ by the first one to the label attributed in $\labin$ by
% the second one. $\tau$ does the same for each output wire, and by definition of
% routing areas $\sigma S \tau^{-1}$ is totally determined and must be equal to
% $R$.

Routing areas may be combined in two ways such that the
resulting proof net reduces to a new routing area. The multirelation
describing the result can be computed directly from the initial multirelations of routing
areas involved, giving a way of building complex circuits from small components.

\paragraph*{Operations}

The first operation, juxtaposition, amounts to put side by side two routing areas.
The result is immediately seen as a routing area itself, described by the
coproduct of the two multirelations:

\begin{definition}{Juxtaposition}\label{definition-mrarea-juxt}\\
  Let $\routarea R = (\labin,\labout,R)$ and $\routarea S =
  (\labin',\labout',S)$, we define the juxtaposition $\routarea R + \routarea S$
  by $(\labin + \labin', \labout + \labout', R + S)$. The corresponding net is
  obtained by juxtaposing the nets of $\routarea R$ and $\routarea S$:

  \begin{center}
    \tikzvcenter{
      \begin{proofnet}
        \node[net={30pt}{30pt}] (R) {R};
        \coordinate (Ri1) at ([yshift=-10pt]R.west);
        \coordinate (Rim) at ([yshift=10pt]R.west);
        \coordinate (Ro1) at ([yshift=-10pt]R.east);
        \coordinate (Ron) at ([yshift=10pt]R.east);
        \init{
          \wirefromto{[xshift=-10pt]Ri1}{[xshift=-10pt]Ri1}{Ri1}{Ri1}
          \wirefromto{[xshift=-10pt]Rim}{[xshift=-10pt]Rim}{Rim}{Rim}
          \wirefromto{[xshift=10pt]Ro1}{[xshift=10pt]Ro1}{Ro1}{Ro1}
          \wirefromto{[xshift=10pt]Ron}{[xshift=10pt]Ron}{Ron}{Ron}
        }
        \pdots[2pt]([xshift=-5pt]Ri1)([xshift=-5pt]Rim)
        \pdots[2pt]([xshift=5pt]Ro1)([xshift=5pt]Ron)
        \node[net={30pt}{30pt},below=10pt of R] (S) {S};
        \coordinate (Si1) at ([yshift=-10pt]S.west);
        \coordinate (Sim) at ([yshift=10pt]S.west);
        \coordinate (So1) at ([yshift=-10pt]S.east);
        \coordinate (Son) at ([yshift=10pt]S.east);
        \init{
          \wirefromto{[xshift=-10pt]Si1}{[xshift=-10pt]Si1}{Si1}{Si1}
          \wirefromto{[xshift=-10pt]Sim}{[xshift=-10pt]Sim}{Sim}{Sim}
          \wirefromto{[xshift=10pt]So1}{[xshift=10pt]So1}{So1}{So1}
          \wirefromto{[xshift=10pt]Son}{[xshift=10pt]Son}{Son}{Son}
        }
        \pdots[2pt]([xshift=-5pt]Si1)([xshift=-5pt]Sim)
        \pdots[2pt]([xshift=5pt]So1)([xshift=5pt]Son)
      \end{proofnet}
      $=$
      \begin{proofnet}
        \node[net={30pt}{70pt}] (box) {\routarea{R} + \routarea{S}};
        \coordinate (i1) at ([yshift=-20pt]box.west);
        \coordinate (im) at ([yshift=20pt]box.west);
        \coordinate (o1) at ([yshift=-20pt]box.east);
        \coordinate (on) at ([yshift=20pt]box.east);
        \init{
          \wirefromto{[xshift=-10pt]i1}{[xshift=-10pt]i1}{i1}{i1}
          \wirefromto{[xshift=-10pt]im}{[xshift=-10pt]im}{im}{im}
          \wirefromto{[xshift=10pt]o1}{[xshift=10pt]o1}{o1}{o1}
          \wirefromto{[xshift=10pt]on}{[xshift=10pt]on}{on}{on}
        }
        \pdots[2pt]([xshift=-5pt]i1)([xshift=-5pt]im)
        \pdots[2pt]([xshift=5pt]o1)([xshift=5pt]on)
      \end{proofnet}
    }
  \end{center}
\end{definition}

% \begin{lemma}{Juxtaposition}\label{lemma-mrarea-juxt}\\
%   Let $\mathbf{R} = (\labin,\labout,R)$ and $\mathbf{S} = (\labin',\labout',S)$
%   be two routing area. Then their juxtaposition is also a routing area
%   $\routarea{R} + \routarea{S}$ defined by $(\labin + \labin',\labout +
%   \labout',R + S)$
%   
%   \end{lemma}

The second one is more involved: the trace operation consists in connecting
an input to an output given that they are not related to begin with, to avoid
the creation of a cycle. Doing so, we remove this output and input from the
external interface, and create new internal paths between remaining inputs and
outputs. If we reduce the resulting net to a normal form, we obtain a new area,
whose multirelation can be computed from the initial one. 

\begin{definition}{Trace}\label{definition-mrarea-trace}\\
  Let $\mathbf{R} = (\labin,\labout,R)$ be a routing area, and $(i,o) \in \labin
  \times \labout$. The \emph{trace} at $(i,o)$ of $\routarea
  R$ is obtained by connecting the input $i$ with the output $o$ of $\routarea
  R$ and reducing this net to its normal form.
\end{definition}

\begin{property}{Trace is a routing area}\label{property-mrarea-trace}\\
  Let $\mathbf{R} = (\labin,\labout,R)$ be a routing area, and $(i,o) \in \labin
  \times \labout$ such that $i \mathrel{R_0} o$. Then the trace at $(i,o)$ of
  $\routarea R$ is a routing area $\routarea T$ :
  \begin{center}
    \tikzvcenter{
      \begin{proofnet}
        \node[net={30pt}{30pt}] (R) {R};
        \coordinate (Ri1) at ([yshift=-10pt]R.west);
        \coordinate (Rim) at ([yshift=10pt]R.west);
        \coordinate (Ro1) at ([yshift=-10pt]R.east);
        \coordinate (Ron) at ([yshift=10pt]R.east);
        \init{
          \wirefromto{[xshift=-10pt]Ri1}{[xshift=-10pt]Ri1}{Ri1}{Ri1}
          \wirefromto{[xshift=-10pt]Rim}{[xshift=-10pt]Rim}{Rim}{Rim}
          \wirefromto{[xshift=10pt]Ro1}{[xshift=10pt]Ro1}{Ro1}{Ro1}
          \wirefromto{[xshift=10pt]Ron}{[xshift=10pt]Ron}{Ron}{Ron}
        }
        \pdots[2pt]([xshift=-5pt]Ri1)([xshift=-5pt]Rim)
        \pdots[2pt]([xshift=5pt]Ro1)([xshift=5pt]Ron)
        \draw[rounded corners=3pt]
          ([xshift=-10pt]Ri1) -- ++(0,-10pt) -| ([xshift=10pt]Ro1);
      \end{proofnet}
      $\to^\ast$
      \begin{proofnet}
        \node[net={30pt}{30pt}] (box) {T};
        \coordinate (i1) at ([yshift=-10pt]box.west);
        \coordinate (im) at ([yshift=10pt]box.west);
        \coordinate (o1) at ([yshift=-10pt]box.east);
        \coordinate (on) at ([yshift=10pt]box.east);
        \init{
          \wirefromto{[xshift=-10pt]i1}{[xshift=-10pt]i1}{i1}{i1}
          \wirefromto{[xshift=-10pt]im}{[xshift=-10pt]im}{im}{im}
          \wirefromto{[xshift=10pt]o1}{[xshift=10pt]o1}{o1}{o1}
          \wirefromto{[xshift=10pt]on}{[xshift=10pt]on}{on}{on}
        }
        \pdots[2pt]([xshift=-5pt]i1)([xshift=-5pt]im)
        \pdots[2pt]([xshift=5pt]o1)([xshift=5pt]on)
      \end{proofnet}
    }
  \end{center}
  Where $\routarea T = (\labin - \{i\},\labout - \{o\},T)$ is defined by the
  multirelation :
  \begin{equation}\label{routing-eq-trace}
    T(x,y) = R(x,y) + R(x,o)R(i,y)
  \end{equation}
\end{property}

The formula~\ref{routing-eq-trace} expresses that in the resulting routing area
$\routarea T$, the total number of ways to go from an input $x$ to an output $y$
is the number of direct paths $R(x,y)$ from $x$ to $y$ that were originally in $R$, plus
all the ways of going from $x$ to $o$ times the ways of going from $i$ to $y$.
Indeed, any pair of such paths yields a new distinct path in the trace once
$i$ and $o$ have been connected.

\begin{proof}Property~\ref{property-mrarea-trace}\\
  If $\arity{i} = \arity{o} = 0$, then we connected a coweakening to a weakening
  and we can erase them to recover the desired area where $T$ is just the
  restriction of $R$ to $(\labin - \{i\}) \times (\labout - \{o\})$, which
  agrees with the formula of \ref{property-mrarea-trace} as the product
  $R(x,o)R(i,y)$ is always zero.

  Now, assume that $\arity{o} = 0 < \arity{i}$. We can reduce the introduced
  redex as follow :
\begin{center}
  \tikzvcenter{
    \begin{proofnet}
      % Tree of i1
      \init[grow=left]{
        \zeroary[normal,name=i1]{\oc}
      }
      % Tree of im
      \init[grow=right]{
        \nary[name=im]{\wn}{
          \wire{\point(im1)},
          \point(im2),
          \wire{\point(im3)}
        }
      }
      \pdots[2pt](i11)(i13)
      \pdots[2pt](im1)(im3)
      % \pdots[5pt](i1.pal)(im.pal)
      % \pdots[2pt](o11)(o13)
      % \pdots[2pt](on1)(on3)
      % \pdots[5pt](o1.pal)(on.pal)
    \end{proofnet}  
    $\to^\ast$
    \begin{proofnet}
      % Tree of o1
      \init[grow=left]{
        \zeroary[normal,name=o1]{\oc}
      }
      % Tree of on
      \init[grow=left]([yshift=7.5ex,xshift=5pt]o1.pal){
        \zeroary[normal,name=on]{\oc}
      }
      \pdots[5pt](o1.pal)(on.pal)
    \end{proofnet} 
  }
\end{center}

These coweakening are connected to the trees of the outputs $\conn{i}$.
These are either connected to a wire, or a
cocontraction tree and we can eliminate superfluous coweakening using the equivalence
relation. Once again, we didn't create new paths and recover an area whose
relation is the restriction of $R$, still agreeing with the formula as
$\arity{i} = 0$ implies $R(x,o)R(i,y)$ being zero again. The dual case $\arity{i} = 0 <
\arity{o}$ is treated the same way.
 
The general case relies on the commutation of contractions and cocontractions
trees that can be derived by iterating the $\red{ba}$ rule. We can apply the
following reduction on the trees of $i$ and $o$ that have been connected :

\begin{center}
  \includegraphics{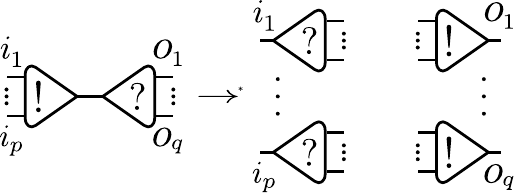}
\end{center}

where wires $i_1,\ldots,i_p$ are connected to the trees of the inputs in
$\conn{o}$ and $o_1,\ldots,o_q$ to the trees of the outputs of $\conn{i}$. The
reduced net now has the shape of a routing area. As before, the direct
paths between $x$ and $y$ when $x \neq i$ and $y \neq o$ are left unchanged. But
any couple of paths in $R$ between $x$ and $o$ arriving at some $i_k$ and
between $i$ and $y$ arriving at some $o_l$ yields exactly one new path between $x$
and $y$ in the new area (see \ref{property-mrarea-path-preservation} below). By
definition, there are $R(x,o)$ paths connecting $x$ to $o$ and $R(i,y)$
connecting $i$ to $y$ : hence there are $R(x,o)R(i,y)$ such couples.
\end{proof}
%%%%%%%%%%%%%%%%%% --- END EXTENDED --- %%%%%%%%%%%%%%%%%

These two operations are sufficient to implement composition which is the
connection of an output of an area to an input of another area. Composition is a
fundamental feature of routing areas. This is what makes them modular, allowing
to build routing areas by composing simple blocks. To connect an output $o$ of
$\routarea R$ to an input $i$ of $\routarea S$, we first perform the
juxtaposition followed by a trace at $(i,o)$.  This is similar both in form and
in spirit, to the composition of Game Semantic or Geometry of Interaction whose
motto is ``composition = parallel composition plus hiding''. 

\begin{corollary}{Composition}\label{corollary-mrarea-composition}\\
  Let $\routarea{R} = (\labin,\labout,R)$ and $\routarea{S} = (\labin',\labout',S)$ be
  two routing areas, $o \in \labout$ and $i \in \labin'$. Then the net resulting
  from connecting the output $o$ to the input $i$ can be reduced to a new
  routing area $\routarea{T} = (\labin + \labin' - \{i\}, \labout + \labout' -
  \{o\}, T)$
  
  \begin{center}
    \tikzvcenter{
      \begin{proofnet}
        \node[net={30pt}{30pt}] (R) {R};
        \coordinate (Ri1) at ([yshift=-10pt]R.west);
        \coordinate (Rim) at ([yshift=10pt]R.west);
        \coordinate (Ro1) at ([yshift=-10pt]R.east);
        \coordinate (Rok) at (R.east);
        \coordinate (Ron) at ([yshift=10pt]R.east);
        \init{
          \wirefromto{[xshift=-10pt]Ri1}{[xshift=-10pt]Ri1}{Ri1}{Ri1}
          \wirefromto{[xshift=-10pt]Rim}{[xshift=-10pt]Rim}{Rim}{Rim}
          \wirefromto{[xshift=10pt]Ro1}{[xshift=10pt]Ro1}{Ro1}{Ro1}
          \wirefromto{[xshift=10pt]Ron}{[xshift=10pt]Ron}{Ron}{Ron}
        }
        \pdots[2pt]([xshift=-5pt]Ri1)([xshift=-5pt]Rim)
        \pdots[2pt]([xshift=5pt]Ro1)([xshift=5pt]Rok)
        \pdots[2pt]([xshift=5pt]Rok)([xshift=5pt]Ron)
        \node[net={30pt}{30pt},right=10ex of R] (S) {S};
        \coordinate (Si1) at ([yshift=-10pt]S.west);
        \coordinate (Sil) at (S.west);
        \coordinate (Sim) at ([yshift=10pt]S.west);
        \coordinate (So1) at ([yshift=-10pt]S.east);
        \coordinate (Son) at ([yshift=10pt]S.east);
        \init{
          \wirefromto{[xshift=-10pt]Si1}{[xshift=-10pt]Si1}{Si1}{Si1}
          \wirefromto{[xshift=-10pt]Sim}{[xshift=-10pt]Sim}{Sim}{Sim}
          \wirefromto{[xshift=10pt]So1}{[xshift=10pt]So1}{So1}{So1}
          \wirefromto{[xshift=10pt]Son}{[xshift=10pt]Son}{Son}{Son}
          \wirefromto{Rok}{Rok}{Sil}{Sil}
        }
        \pdots[2pt]([xshift=-5pt]Si1)([xshift=-5pt]Sil)
        \pdots[2pt]([xshift=-5pt]Sil)([xshift=-5pt]Sim)
        \pdots[2pt]([xshift=5pt]So1)([xshift=5pt]Son)
      \end{proofnet}
      $\to^\ast$
      \begin{proofnet}
        \node[net={30pt}{30pt}] (box) {T};
        \coordinate (i1) at ([yshift=-10pt]box.west);
        \coordinate (im) at ([yshift=10pt]box.west);
        \coordinate (o1) at ([yshift=-10pt]box.east);
        \coordinate (on) at ([yshift=10pt]box.east);
        \init{
          \wirefromto{[xshift=-10pt]i1}{[xshift=-10pt]i1}{i1}{i1}
          \wirefromto{[xshift=-10pt]im}{[xshift=-10pt]im}{im}{im}
          \wirefromto{[xshift=10pt]o1}{[xshift=10pt]o1}{o1}{o1}
          \wirefromto{[xshift=10pt]on}{[xshift=10pt]on}{on}{on}
        }
        \pdots[2pt]([xshift=-5pt]i1)([xshift=-5pt]im)
        \pdots[2pt]([xshift=5pt]o1)([xshift=5pt]on)
      \end{proofnet}
    }
  \end{center}
\end{corollary}

\begin{remark}
This operation can be generalized to the connection of $n$ outputs of $R$ to
$n$ inputs of $S$. When $n = \cardn{\labout} = \cardn{\labin'}$, the 
multirelation $T$ describing the resulting routing area is the
composed $S \circ R$.
\end{remark}

The following property gives the high level operational behavior of a routing
area. It supports our interpretation of routing areas as dispatchers of
exponential boxes. Given a closed exponential box, we connect it to the
auxiliary port of a cocontraction to obtain a module which can then be connected
to an input $i$ of an area.  Through reduction, the box will traverse the area
and be duplicated $R(i,o)$ times to each output $o$. The role of the additional
cocontraction is to preserve the area and allow future connections to the same
input. We would get a similar transit property connecting directly the
exponential box to $i$, but the process is destructive as it erases
the input wire of $i$ and prevents any future use.

\begin{property}{Transit}\label{property-translation-transit}\\
  Let $\sigma$ be a closed exponential box, $\mathbf{R} = (\labin,\labout,R)$ a
  routing area, $i \in \labin$. Let $\{ o_1, \ldots, o_p \} = \conn{i}$ and for $1 \leq k \leq p,\ c_k = R(i,o_k)$. Then :
  \begin{center}
    \includegraphics{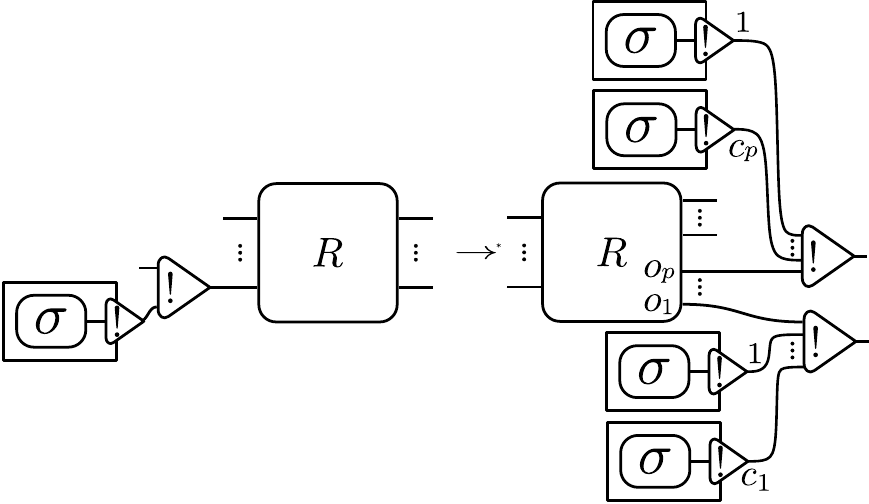} 
  \end{center}
\end{property}

The property is straightforward application of reduction rules and the net
equivalence.

\subsection{The Routing Semantic}

Routing areas do not only fulfill practical needs. They are general enough to be
the language of normal forms of \emph{routing nets}. Let us first give a precise
definition. A correct net is a net satisfying the correctness criterion defined
in~\cite{ehrhard:hal-00150274}:

\begin{definition}{Routing nets}\label{definition-mrarea-rnets}\\
  A routing net $\net R$ is a correct net composed only of weakenings,
  coweakenings, contractions, cocontractions, and possibly floating wires.
  Moreover, we ask that all wires are labelled with the same formula $\oc A$,
  fixing de facto their orientation.
\end{definition}

%%%%%%%%%%%%%%%%%%%% --- EXTENDED --- %%%%%%%%%%%%%%%%%%%
Paths will also be of interest in the rest of this subsection. 

\begin{definition}{Paths}\label{definition-mrarea-paths}\\
  Let $\net R$ be a net. We recall that for a cell $c$ of $\net R$, we write
  $p_i(c)$ for its $i$-th auxiliary port (if it exists) and $p(c)$ for its main
  port. For a wire $w$, $s(w)$ designates the source port of $w$ while $e(w)$
  is its end port.
  Let $\net R$ be a net without exponential boxes, we construct the associated
  undirected graph $G(\net R)$ with ports as vertices and: 
  \begin{description}
    \item[Wire edges] For any wire $w$ of $\net R$, we add an edge between $s(w)$
      and $e(w)$.
    \item[Cell edges] For every auxiliary port $p_i(c)$ of a cell, we add an edge between $p_i(c)$ and $p(c)$.
  \end{description}
  A path $p$ in $\net R$ is a finite sequence
  $(p_1,e_1,p_2,e_2,\ldots,e_n,p_{n+1})$ such that $p_i$ is a port of $\net R$,
  $e_i$ an edge of $G(R)$ linking $p_i$ and $p_{i+1}$ and such that $e_i$ and
  $e_{i+1}$ are of distinct nature (cell/wire edge). We extend $s$ and $e$ to operate on path, defined by
  $s(p) = p_1$ and $e(p) = p_{n+1}$. Paths whose starting and ending edges are wire edges
  can also be described as a sequence of corresponding wires $(w_1,\ldots,w_m)$
  as internal ports and cell edges can be recovered.  We note $P(\net R)$ the
  set of paths in $\net R$ and $P_f(\net R)$ the paths starting and ending on
  free ports. 
\end{definition}

Albeit closed to \emph{switching paths} (the ones involved in the correctness
criterion), they do not match exactly. A path in a switching graph can arrive at
an auxiliary port of a cocontraction and bounce back in the other, while the
paths we have defined here must continue via the principal port. However, a
correct (routing) net do not contain cyclic paths, which means that switching acyclicity
implies acyclicity:

\begin{property}{Acyclicity of correct nets}\label{property-mrarea-path-switching}\\
  A routing net is acyclic, that is there is no path $p$ such that $s(p) = t(p)$.
\end{property}

\begin{proof}\ref{property-mrarea-path-switching}\\
  We prove that the existence of a cycle implies the existence of a cycle in a
  switching graph. Assume that a cycle exist and take a one of minimal
  length.  Assume that the two auxiliary ports of a cell (contraction or
  cocontraction) are both visited by the cycle. By the definition of paths, when
  the cycle visits one of these auxiliary port, it must continue in the
  principal port.  But then a smaller cycle would be derivable, by taking the
  sub path starting at the principal port of this cell and stopping as long as
  it comes back by any of the auxiliary ports. Thus a minimal cycle visits at most
  once of the auxiliary port of any cell, and is also a cycle in a switching
  graph.  
\end{proof}

We now state a fundamental property relating reduction to paths:

\begin{property}{Path preservation}\label{property-mrarea-path-preservation}\\
  Let $\net R$ be a routing net, $p_1$ and $p_2$ be free ports. We write
  $P_f(\net R,p_1,p_2) = \{
  p \in P_f(\net R) \mid s(p) = p_1, e(p) = p_2 \}$. If $\net R \to
  \net{R}'$, then $P_f(\net R,p_1,p_2) = P_f(\net{R}',p_1,p_2)$. In particular,
  $P_f(\net R) = \sum P_f(\net R,p_i,p_{i'}) = P_f(\net{R}')$.
\end{property}

\begin{proof}\ref{property-mrarea-path-preservation}
  The proof is simple in the case of routing nets. As a path in $P_f(R,p_1,p_2)$
  must begin and end in a free port, it can't go through a weakening or a
  coweakening, and the corresponding reductions leave such paths unchanged. The
  only relevant reduction is the $\red{ba}$ rule:
  \begin{center}
    \includegraphics{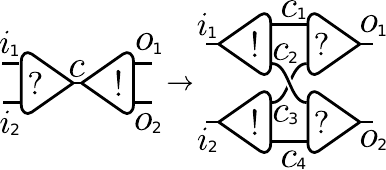}
  \end{center}
  Let us define the application $\tau : P_f(R,p_1,p_2) \to P_f(R',p_2,p_2)$. For
  a path $p$ which does not cross the redex, $\tau(p) = p$. Otherwise, we
  replace any subsequence in the left column by the one in the right column:
  \begin{center}
    \begin{tabular}{c|c}
      Subsequence & Image by $\tau$ \\
      \hline
      $i_1,c,o_1$ & $i_1,c_1,o_1$ \\
      $i_1,c,o_2$ & $i_1,c_2,o_2$ \\
      $i_2,c,o_1$ & $i_2,c_3,o_1$ \\
      $i_2,c,o_2$ & $i_2,c_4,o_2$ \\
    \end{tabular}
  \end{center}
  We omitted the four other possibilities which can be deduced from this table
  by reversing both the subsequence and its image. It is easily seen that $\tau$
  is an bijection.
\end{proof}
%%%%%%%%%%%%%%%%%% --- END EXTENDED --- %%%%%%%%%%%%%%%%%

The following theorem establishes the link between routing
areas and normal forms of routing nets. We propose two different intuitive
explanations of why the theorem holds:

\begin{enumerate}
  \item The basic components of routing nets, (co)contractions, wires
    and (co)weakenings, are routing areas. Then, juxtaposition and trace
    operations are general enough to combine them into an arbitrary routing net that
    reduces to a routing area according to Property~\ref{property-mrarea-trace}.
  \item In a routing net, the $\red{ba}$ rule allows to commute all contractions
    and cocontractions. Then, by equivalence, we can erase weakenings and
    coweakenings that are not connected to a free wire. At the end of this
    process, the resulting net must have the shape of a routing area.
\end{enumerate}

\begin{theorem}{Routing area characterization}\label{theorem-mrarea-charact}\\
 The normal form of a routing net $\net S$ is a routing area $\routarea R =
 (\labin,\labout,R)$.
\end{theorem}

For a routing net $\net S$, we can define the application $\sem{.} : \net S
\mapsto R$ that maps $S$ to the multirelation $R$ describing its normal form. By
unicity of normal forms, this application is invariant by reduction and is thus a
semantic for routing nets. It has the following properties:
\begin{description}
  \item[Soundness] The multirelation only depends on the normal form.
  \item[Adequacy] Two routing nets with the same denotation have the same normal
    form because a multirelation defines a routing area uniquely.
  \item[Full completeness] Any multirelation between finite sets is realized by the 
    associated routing area.
  \item[Compositionnality] We can compute the semantic of a net in a
    compositional way from the semantic of its smaller parts through
    juxtaposition and trace.
\end{description}

%%%%%%%%%%%%%%%%%%%% --- EXTENDED --- %%%%%%%%%%%%%%%%%%%
\begin{proof}Theorem~\ref{theorem-mrarea-charact}\\
 We prove the result by induction on the number $n$ of cells of $\net R$.
 \begin{itemize}
   \item ($n=0$) $\net R$ is only composed of free wires. We take $\labin = \{
     s(w) \mid w \mbox{ wire }\}$, $\labout = \{ e(w) \mid w \mbox{ wire } \}$
     and $R$ is the relation defined by $i \mathrel R o \iff \exists w,\ i =
     s(w), o = e(w)$.
   \item (induction step) Let take any node $N$ of $\net R$. We call
     $\net{R}'$ the subnet obtained by removing $N$ and replacing its ports by
     free ports.

    \begin{center}
      \includegraphics{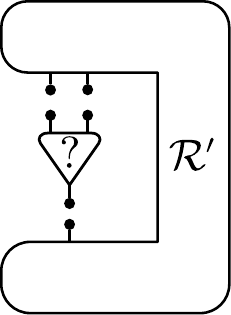}
    \end{center}

    By induction, $\net{R}'$ can be reduced to a routing area
    $\routarea{R}'$.

     \begin{description}
       \item[(co)weakening] If $N$ is a weakening or coweakening, it is a
         routing area and can be composed with $\routarea{R}'$, and reduced to a
         new routing area $\routarea{R}$ according to
         \ref{corollary-mrarea-composition}. Thus we can
         reduce the whole net to $\routarea{R}$. 
       \item[(co)contraction] If $N$ is a contraction or a cocontraction, it can
         still be seen as a routing area and we juxtapose it to $\routarea{R}'$.
         What remains to do is to perform three traces operations to recover the
         original net and reduce the whole net to a routing area
         $\routarea{R}$. However, we must ensure that the input and the output
         we connect at each step are not already connected in the routing
         area. The first operation is always legal, as it is actually a merge
         operation of previously disjoint areas. The following traces are also
         valid, as reduction does not create cycles, as implied by
         \ref{property-mrarea-path-preservation}. 
     \end{description}
 \end{itemize}
\end{proof}
%%%%%%%%%%%%%%%%%% --- END EXTENDED --- %%%%%%%%%%%%%%%%%

For a routing net $\net S$, we can define the application $\sem{.} : \net S
\mapsto R$ that maps $S$ to the multirelation $R$ describing its normal form. By
unicity of normal forms, the application is invariant by reduction and is thus a
semantic for our routing nets, with the following remarkable properties:
\begin{description}
  \item[Sound] The multirelation only depends on the normal form, this is invariant by reduction.
  \item[Adequate] Two routing nets with the same denotation have the same normal
    form, as a multirelation define a routing area uniquely.
  \item[Fully complete] Any multirelation on finite sets is realised by the 
    associated routing area.
  \item[Compositionnal] We can compute the semantic of a net in a
    compositional way from the semantic of its smaller parts, through
    juxtaposition and trace.
\end{description}

%%%%%%%%%%%%%%%%%%%% --- EXTENDED --- %%%%%%%%%%%%%%%%%%%
This semantic can be defined without resorting reduction, by
counting paths.

\begin{theorem}{Path semantic}\label{theorem-mrarea-pathsemantic}\\
  Let $\net S$ be a routing net. Let $\labin$ be the set of free ports of
  $\net S$ which are the source point of a wire, and $\labout$ the ones
  that are the end point. Then $\sem{\net S}$
  is the multirelation between $\labin$ and $\labout$ given by
  \[
    \sem{\net S}(i,o) = \cardn{ P_f(\net S,i,o) }
  \]
\end{theorem}

\begin{proof}\ref{theorem-mrarea-pathsemantic}
  \begin{itemize}
    \item For normal forms, this follows from the very definition of a
      routing area associated to the multirelation $\sem{\net S}$ :
      $\sem{\net S}(i,o)$ counts precisely the number of paths starting
      at port $i$ and ending at port $o$.
    \item Otherwise, by path preservation (\ref{property-mrarea-path-preservation}), this
      quantity is invariant by reduction. We conclude by induction on the length
      of the longest reduction of $\net S$ to its normal form.
  \end{itemize}
\end{proof}

\begin{remark}{Communication Areas}\label{mrem-mrarea-communication-areas}\\
The communication areas defined in~\cite{EHRHARD2010606} are a special case of
routing areas: for $n \leq 1$, the $n$-communication area is the routing area
$(\{1,\ldots,n\},\{1,\ldots,n\},R)$ where $x \mathrel{R} y \iff x \neq y$.
\end{remark}

\begin{remark}{Canonicity of multirelations}\\
  The multirelation defining an area is actually not unique from a set-theoretic
  point of view: indeed, for $\routarea R = (\labin,\labout,R)$, then all
  multirelations in $\{ \tau^{-1} \circ R \circ \sigma : E \to F \mid \sigma : E
    \to \labin, \tau : F \to \labout, \sigma \text{ and } \tau \text{ bijective
  } \}$ describe the same area. Even though this object is a proper class,
  all these relations are considered isomorphic (for example, in the arrow
  category of \cat{FMRel}). After all, even finite deterministic automata suffer
  this kind of subtlety, which is not relevant in practice.
\end{remark}

We are now armed to encode communication primitives.  This is illustrated in
Section~\ref{section:translation} by the translation of the {\lamadio} calculus,
succinctly described in the following section.

\section{The Concurrent $\lambda$-calculus \lamadio}
\label{section:lamadio}

We present the 
{\lamadio}~\cite{Amadio2009} calculus, 
following the presentation in Madet's Ph.D
thesis~\cite{madet:tel-00794977}. 
%%%%%%%%%%%%%%%%%%%%%%%%%%%%%%%%%%%%%%%%%%%%%%%%%%%%%%%%%%%%%%%%%%%%%%

The \lamadio\ calculus is a call-by-value $\lambda$-calculus, equipped with {\em
references} that abstract the notion of global memory cells. The calculus is
enriched with a parallel composition operator $\parallel$ for modeling
concurrency.

\subsection{Syntax}
Variables are
denoted with $x,y,\ldots$ while references are
denoted with $r,s,\ldots$. The language consists of {\em values}, {\em
  terms}, {\em stores} and {\em programs}. A {\em store} is a
top-value set of associations between references and values, while a
program is a store together with a set of terms. The terms in a
program can be regarded as threads running in parallel, while the
stores represent the state of the global memory. Programs are the
objects of interest in \lamadio.
\begin{center}
\begin{tabular}{llll}
  -values & $V$ & $::=$ & $x \mid \ast \mid \lambda x.M $ \\
  -terms & $M$ & $::=$ & $V \mid M\ M \mid \get{r} \mid
  \set{r}{V} \mid M \parallel M$ \\
  -stores & $S$ & $::=$ & $r \Leftarrow V \mid (S \parallel S)$ \\
  -programs & $P$ & $::=$ & $M \mid S \mid (P \parallel P)$ \\
\end{tabular}
\end{center}
The constant value $\ast$ stands for the return value of a reference
assignation: it carries no particular information.
The primitives $\set{r}{V}$ and $\get{r}$ respectively
writes a value $V$ to and reads from a given reference $r$.
While {\setw}s can be only performed on values, a more general
$\set{r}{M}$ can be encoded as $(\lambda x . \set{r}{x})\,M$ for an
arbitrary term $M$. 
Once reduced, a $\set{r}{V}$ produces the special kind of thread
$r \Leftarrow V$ at top level in a store, adding $V$ as the
possible values available for the reference $r$.
Parallelism is accounted for using the operator~$\parallel$. Terms,
programs and stores can be placed in parallel. For stores, this simply
means that all the corresponding associations reference/value are
available for substitution of reference in the threads.

%%%%%%%%%%%%%%%%%%%%%%%%%%%%%%%%%%%%%%%%%%%%%%%%%%%%%%%%%%%%%%%%%%%%%%
\subsection{Reduction}

The calculus is endowed with the usual structural rules for the
parallel operator, namely associativity and commutativity
(Table~\ref{fig:lam-struct-rules}). 
The rewrite rules for the language are found in
Table~\ref{fig:lam-red-rules}. Together with the $\beta_v$ rule are
two new reductions, one that turns an assignment to a store and one
that turns a get to a value. 
The rules use the evaluation contexts defined in
Table~\ref{fig:lam-eval-context} to handle congruence. The context $E$ denotes a
weak call-by-value weak reduction which is neither right-to-left or
left-to-right. The context $C$ allows reduction to occurs in any thread of a
program.

The substitution of a reference is a non-deterministic
operation. Reference must be seen as an abstraction for a set of typed memory
cells that can hold many values. Let $\mbox{proj}_1 = \lambda x y . x$ and
$\mbox{proj}_2 = \lambda x y . y$ be the Church projections, and consider the
term $$P = (\lambda x . x\ V_1\ V_2)\ \get{r} \parallel (\lambda y .\set{r}{y})\
\get{s} \parallel \set{s}{\mbox{proj}_1} \parallel \set{s}{\mbox{proj}_2}$$ 
The two \setw\ can be reduced to stores :
$$P \to^\ast (\lambda x . x\ V_1\ V_2)\ \get{r} \parallel (\lambda y . \set{r}{y})\
\get{s} \parallel \ast \parallel \ast \parallel \store{s}{\mbox{proj}_1}
\parallel \store{s}{\mbox{proj}_2}$$ for some distinct values $V_1$ and $V_2$.
Here, $\get{s}$ have essentially two incompatible ways to reduce : either it
is replaced by $\mbox{proj}_1$ or by $\mbox{proj}_2$. In the first case,
\begin{tabred}
  P & \to^\ast & (\lambda x. x\ V_1\ V_2)\ \get{r} \parallel \set{r}{\mbox{proj}_1}
  \parallel P' \\
  & \to^\ast & (\lambda x . x\ V_1\ V_2)\ \get{r} \parallel \ast \parallel
  \store{r}{\mbox{proj}_1} \parallel P' \\
  & \to^\ast & (\lambda x . x\ V_1\ V_2)\ \mbox{proj}_1 \parallel \ast \parallel
  \store{r}{\mbox{proj}_1} \parallel P' \\
  & \to^\ast & V_1 \parallel \ast \parallel
  \store{r}{\mbox{proj}_1} \parallel P' \\
\end{tabred}
  where $P' = \ast \parallel \ast \parallel \store{s}{\mbox{proj}_1} \parallel
  \store{s}{\mbox{proj}_2}$. However if $\get{s}$ is replaced by
  $\mbox{proj}_2$, we get that $P \to^\ast V_2 \parallel \ast \parallel
  \store{r}{\mbox{proj}_2} \parallel P'$ : $P$ has two distinct normal forms.

Despite the $\parallel$ operator being a static constructor, it can be embedded
in abstractions and thus dynamically liberated or duplicated. For example, the
term $(\lambda f. f\ \ast \parallel f\ \ast)$ act like a \emph{fork} operation:
if applied to $M$, it generates two copy of its argument in two parallel threads
$M\ \ast \parallel M\ \ast$.
The next section is devoted to detailing how terms of {\lamadio} are translated
to proof nets described in \ref{section:nets}, thanks to the areas introduced in
\ref{section:areas}.
%%% ---------------------------------------------------------------------
\begin{table}[t]
\centering
\begin{minipage}{.4\linewidth}
\centering
\begin{tabular}{rcl}
$P \parallel P'$ & = & $P \parallel P'$ \\
$(P \parallel P') \parallel P''$ & = & $P \parallel (P' \parallel
P'')$\mynl
\end{tabular}
\caption{Structural Rules}\label{fig:lam-struct-rules}
\end{minipage}
\begin{minipage}{.4\linewidth}
\centering
\begin{tabular}{lcl}
$E$ & $::=$ & $[.] \mid E\ M \mid M\ E$ \\
$C$ & $::=$ & $[.] \mid (C \parallel P) \mid (P \parallel C)$\mynl
\end{tabular}
\caption{Evaluation Contexts}\label{fig:lam-eval-context}
\end{minipage}
\mynl
\begin{tabular}{lrcl}
$(\beta_v)$ & $C[E[(\lambda{x}.M)\ V]$ &
  $\rightarrow$ & $ C[E[M[V/x]]]$ \\
$(\texttt{\text{get}})$ & $C[E[\get{r}]] \parallel r \Leftarrow V$ &
  $\rightarrow$ & $ C[E[V]] \parallel r \Leftarrow V$ \\
$(\texttt{\text{set}})$ & $C[E[\set{r}{V}]]$ &
  $\rightarrow$ & $ C[E[\ast]] \parallel r \Leftarrow V $\mynl
\end{tabular}
\caption{Reduction Rules}\label{fig:lam-red-rules}
\end{table}
%%%-------------------------------------------------------------------

%
% \subsection{Terms}
%
% \begin{tabular}{llll}
%   -variables & \multicolumn{3}{l}{$x,y,\ldots$} \\
%   -regions & \multicolumn{3}{l}{$r,s,\ldots$} \\
%   -values & $V$ & $::=$ & $x \mid \ast \mid \lambda x.M $ \\
%   -terms & $M$ & $::=$ & $V \mid M\ M \mid \get{r} \mid
%   \text{set}(r,V) \mid M \parallel M$ \\
%   -stores & $S$ & $::=$ & $r \Leftarrow V \mid (S \parallel S)$ \\
%   -programs & $P$ & $::=$ & $M \mid S \mid (P \parallel P)$ \\
% \end{tabular}
%
% \subsection{Reduction}
%
% \begin{center}
% \begin{tabular}{rcl}
% \hline
% \multicolumn{3}{|c|}{Structural rules} \\
% \hline
% $P \parallel P'$ & = & $P \parallel P'$ \\
% $(P \parallel P') \parallel P''$ & = & $P \parallel (P' \parallel P'')$ \\
% \end{tabular}
% \end{center}
%
% \paragraph{}
% \begin{center}
% \begin{tabular}{lcl}
% \hline
% \multicolumn{3}{|c|}{Evaluation contexts} \\
% \hline
% $E$ & $::=$ & $[.] \mid E\ M \mid V\ E$ \\
% $C$ & $::=$ & $[.] \mid (C \parallel P) \mid (P \parallel C)$ \\
% \end{tabular}
% \end{center}
%
% \paragraph{}
% \begin{center}
% \begin{tabular}{lrcl}
% \hline
% \multicolumn{4}{|c|}{Reduction rules} \\
% \hline
% $(\beta_v)$ & $C[E[(\lambda{x}.M)\ V]$ &
%   $\rightarrow$ & $ C[E[M[V/x]]]$ \\
% $(\texttt{\text{get}})$ & $C[E[\get{r}]] \parallel r \Leftarrow V$ &
%   $\rightarrow$ & $ C[E[V]] \parallel r \Leftarrow V$ \\
% $(\texttt{\text{set}})$ & $C[E[\text{set}(r,V)]]$ &
%   $\rightarrow$ & $ C[E[\ast]] \parallel r \Leftarrow V $ \\
%
% \end{tabular}
% \end{center}

% \input{parts/language-ces}

\section{Translation}
\label{section:translation}

To implement the translation, we make use of two specific routing areas that we
introduce below. In the following, $E_i = \{1,\ldots,i\}$ and $R_i$ is the binary relation defined on
$E_i$ by $k \mathrel{R_i} l \iff k \neq l$.

The $\gamma$ area is defined by $(E_3,E_3,R_\gamma = R_3)$. $\gamma$ is actually
a communication area, composed of $3$ pairs of input and outputs grouped by
label. Each such pair represents a plug to which translated terms will be
connected. The definition of $R_\gamma$ expresses that the input and the output
of a plug are not connected, as a component should not receive the data it sent
himself: this would be the analog of a short-circuit. All others inputs and
output are connected.

The $\delta$ area is an analog structure with $4$ plugs: $(E_4,E_4,R_\delta)$.
It is designed to handle the application $M\ N$ which includes three potential sources of effects : 
\begin{enumerate}
  \item The effects $e_1$ produced by reducing $M$ to $\lambda x. M'$
  \item The effects $e_2$ produced by reducing $N$ to value $V_N$
  \item The effects $e_3$ produced by reducing $M'[V_N/x]$ to the final result
    $V$
\end{enumerate}
The reduction of {\lamadio} imposes that $e_1$ and $e_2$ happen before $e_3$,
while $e_1$ and $e_2$ may happen concurrently. For $1 \leq i \leq 3$, the plug
$(i,i)$ of $\delta$ corresponds to the effects $e_i$. The last one is the
external interface for future connections. We easily accommodate $\delta$ to
implements the sequentiality constraint by removing the couples $(3,1)$ and
$(3,2)$ from $R_4$ to form $R_\delta$. Indeed $e_1$ and $e_2$ happens before $e_3$ thus can not observe
any {\setw} made by the latter. Thus we just cut the corresponding wires.  We see
that the formalism of routing areas allows us to easily encode the order of effects.

\subsection{Translating types and effects}
Before translating terms, we need to translate the types from the type and
effects system for {\lamadio} to plain {\LL} formulas. We use the approach of
\cite{tranquilli:hal-00465793}, a monadic translation, explained in the following. Before
translating to {\LL}, let us first try to take a type with effects and translate
it to a pure simple type. Let $e = \{r_1,\ldots,r_n\}$ be an effect (a finite set of references), assume we can assign a simple type $R_i$ to each
reference $r_i$. We can type a store $S_e = R_1 \times \ldots \times R_n$
representing the current state of the memory. We transform a term $M$ of type $A$
producing effects to a pure term which takes the initial state of the store, and
returns the value it computes together with the new state of the store after
this computation. Using curryfication for the arrow type, we define the 
translation: 

\begin{center}
  \begin{tabular}{rcl}
    $T_e(\alpha)$ & $=$ & $S \to S \times \alpha$ \\
    $T_e(A \sur{\to}{e} \alpha)$ & $=$ & $A \to (S \to (S \times \alpha))
    \cong A \times S \to S \times \alpha$
  \end{tabular}
\end{center}

From there, we go to {\LL} types by implementing the pair type $A \times B$ as
$\oc A \otimes \oc B$, and the usual call-by-value translation for the arrow $(A
\to B)^\bullet = \oc(A^\bullet \multimap B^\bullet)$~\cite{MARAIST1995370}. We still
have to determine each $R_i$ first. Using the previous formula, we may
associate an {\LL} type variable $X_{r_i}$ to each reference and plug everything
together to obtain the following equations (where $A_i$ is the type given 
to $r_i$ by the reference context):
\begin{center}
  \begin{tabular}{rcl}
    $\text{Unit}^\bullet$ & $=$ & $\oc{1}$ \\
    $(A \sur{\to}{\{s_1, \ldots, s_m\}} \alpha)^\bullet$ & $=$ 
     & $\oc{((A^\bullet \otimes X_{s_1} \ldots \otimes X_{s_m}) \multimap (X_{s_1}
       \otimes \ldots \otimes X_{s_m} \otimes \alpha^\bullet))}$ \\
   $X_{r_i}$ & $=$ & ${A_i}^\bullet$
  \end{tabular}
\end{center}

This system is solvable precisely because the type system is
stratified~\cite{tranquilli:hal-00465793}, and we can thus translates all the
types of {\lamadio} to plain {\LL} types. The behavior type $\mathbf{B}$ will be
translated to types of the form $A_1 \parr \ldots \parr A_n$ as the translation
remembers the types of each threads.

\subsection{Translating terms}

The general form of the translation of a term $x_1:A_1,\ldots,x_n:A_n \vdash M :
(\alpha, \{r_1,\ldots,r_k\})$ is given by 
$$
\tikzvcenter{
  \begin{tikzpicture}
    \node[net={80pt}{40pt}] (delta) {M^\bullet};
    \draw [arrowed] (delta.east) -- node[above,m] {\alpha^\bullet} ++(15pt,0);
    \draw [arrowed] ([xshift=-15pt]delta.south east) coordinate (a) -- node[left,m] {\oc X_{r_k}} ++(0,-15pt) coordinate (b)  node [below,m,dist=1.5pt] (rk) {r_k};
    \draw [arrowed] ([xshift=-40pt]a) coordinate (b) -- node[left,m] {\oc X_{r_k}} ++(0,-15pt)  node [below,m,dist=1.5pt] (r1) {r_1};
    \proofdots(r1.mid east)(rk.mid west)
    \draw [reverse arrowed] (a|-delta.north) -- node[left,m] {\oc X_{r_k}} ++(0,15pt)  node [above,m,dist=1.5pt] (rk) {r_k};
    \draw [reverse arrowed] (b|-delta.north) coordinate (c) -- node[left,m] {\oc X_{r_1}} ++(0,15pt)  node [above,m,dist=1.5pt] (r1) {r_1};
    \proofdots(r1.mid east)(rk.mid west)
    \draw [reverse arrowed] ([yshift=-10pt]delta.north west) -- node[above,m] (lab1) {A_1^\bullet} ++(-15pt,0) node[left,m,dist=1pt] (x1) {x_1};
    \draw [reverse arrowed] ([yshift=10pt]delta.south west) -- node[below,m] (lab2) {A_n^\bullet} ++(-15pt,0) node[left,m,dist=1pt] (xn) {x_n};
    \proofdots(x1)(xn)
  \end{tikzpicture}
}
$$

We distinguish three different types of free wires:
\begin{description}
  \item[Output wire] The right wire, labelled by $\alpha^\bullet$, corresponds
    to the result of the whole term.
  \item[Variable wires] Each wire on the left corresponds to a variable of the
    context. The (explicit) substitution of a variable $x$ for a term
    $V^\bullet$ is obtained by connecting the output wire of $V^\bullet$ to the
    wire of $x$.
  \item[References wires] The wires positioned at the top are
    input wires corresponding to references and have a similar role as variable wires,
    while the wires at the bottom corresponds dually to the output. References
    wires will be connected by routing areas.
\end{description}
    
We present some representative cases of the translation: $\get{r}$ and
$\set{r}{V}$ for reference management, the abstraction to show how effects are
thunked in a function's body following the monadic translation, and the
application that shows the usage of routing areas to handle non-trivial effects
scheduling. We start with reference operations, which serve as switch from and
to reference wires:

% $$
% \tikzvcenter{
%   \tikztarget(mtrad){
%     \begin{proofnet}[baseline=-5.ex]
%       \init[grow=left]{
%         \boxpal{lambda}{
%           \binary{\parr}{
%             \binary{\otimes}{
%               \nary[name=mtrad,shape=rectangle,cell={net={40pt}{60pt}}]{M}{
%                 {\wire{}},
%                 {\point(xvar)}
%               }}
%               {\point(outeff)}}
%             {
%               \binary{\parr}{}{\point(ineff)}
%             }
%         }
%       }
%       \draw (outeff) -- (mtrad.east);
%       \drawbox{lambda}
%     \end{proofnet}
%   }
% }
% $$

\paragraph*{Set (Figure~\ref{figure:set})}
A $\set{r}{V}$ connects the output of the translation of $V$ to the output
reference wire corresponding to $r$. As other {\setw s} are not relevant a
weakening is connected on the input wire to ignore any incoming resource. In
{\lamadio}, a set reduces to $\ast$ as it does not compute anything valuable.
Consequently the output is the conclusion of a banged $1$ which is the
translation of $\ast$.

One important remark is that an additional exponential layer is added around the
translation of $V$. In a call-by-value language, the non determinism is strict
in the sense that non-deterministic term must be evaluated before any copy. For
example, the term $(\lambda fx .f\ x\ x)\ \get{r} \parallel \store{r}{V_1}
\parallel \store{r}{V_2}$ can reduce either to $f\ V_1\ V_1$ or $f\ V_2\ V_2$
but not to $f\ V_1\ V_2$. Differential {\LL} rather implements the latter
call-by-name semantic as hinted at by the $\red{ba}$ rule which 
expresses that duplication and non-determinism should commute. The mismatch is
due to two different usages we want to make of the $\oc$: 

\begin{itemize}
  \item The first one allows to discriminate what proof nets can be the target
    of structural rules, which implements substitution. In call-by-value, the
    only terms that can be substituted are values. The $\oc$ is introduced by
    the translations of values, using $\oc_p$, and eliminated at usage - when
    applied to another term - for each copy by a dereliction.
  \item The second usage relates to the differential part. The bang denotes
    resources that may be packed non deterministically by a cocontraction. The
    choice is made when a dereliction is met.
\end{itemize}

But as we noted, these two usages are in contradiction: a non-deterministic
packing should not be allowed to be substituted. Technically, the
dereliction corresponding to the place of usage and the dereliction corresponding
to the non-deterministic choice should not be the same. This is the reason
of the additional $\oc$ layer introduced by an exponential box around
$V^\bullet$. The corresponding dereliction is found in the translation of $\get{r}$.

\begin{figure}[tb]
\centering
\begin{minipage}{0.29\textwidth}
\centering
  $\infer[(set)]
  { R;\Gamma \vdash set(r,V) : (\text{Unit},\{ r \}) }
  { R;\Gamma \vdash r : \text{Reg}_r{A} \qquad R;\Gamma \vdash V : (A,\emptyset)
  }$
\end{minipage}
\begin{minipage}{0.69\textwidth}
\centering
  \includegraphics{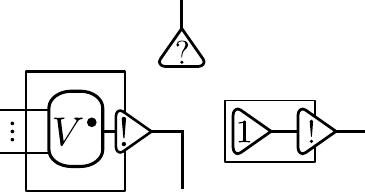}
\end{minipage}
\caption{Translation of $\set{r}{V}$}\label{figure:set}
\end{figure}

\paragraph*{Get (Figure~\ref{figure:get})}
The $\get{r}$, dual of the set, takes a resource from the corresponding input reference wire
and redirects it to the output wire. It outputs a coweakening on the reference wire
as it does not produce any {\setw}. As mentioned in the previous case, a
dereliction is added on the input wire to force the non-deterministic choice and
strip the exponential layer added by the set.

\begin{figure}[tb]
\centering
\begin{minipage}{0.29\textwidth}
\centering
  $\infer[(get)]
  { R;\Gamma \vdash \get{r} : (A,\{ r \}) }
  { R \vdash \Gamma }$
\end{minipage}
\begin{minipage}{0.69\textwidth}
\centering
 \includegraphics{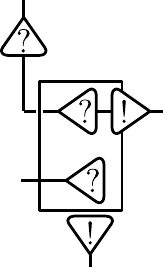}
\end{minipage}
\caption{Translation of $\get{r}$}\label{figure:get}
\end{figure}

\paragraph*{Abstraction (Figure~\ref{figure:abs})}
The abstraction thunks the potential effects of the body $M$ in the pure term
$\lambda x. M$. Following the monadic translation, the input effects are
tensorized with the bound variable, and the output effects with the output of
$M$. Finally the whole term is put in an exponential box as it is a value.

\begin{figure}[tb]
\centering
\begin{minipage}{0.29\textwidth}
\centering
  $\infer[(lam)]
    { R;\Gamma \vdash \lambda{x}.M : (A \sur{\to}{e} \alpha, \emptyset) }
    { R;x:A,\Gamma \vdash M : (\alpha,e) }$
\end{minipage}
\begin{minipage}{0.69\textwidth}
\centering
  \includegraphics{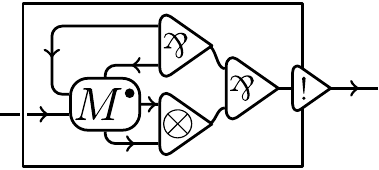}
\end{minipage}
\caption{Translation of $\lambda x. M$}\label{figure:abs}
\end{figure}

\paragraph*{Application (Figure~\ref{figure:app})}
Finally, the application put the routing area at use. Using the same terminology
as in the introduction of this section, we see the effects $e_1$ and $e_2$
coming respectively from the evaluation of $M$ and $N$, and $e_3$, liberated by
the body of the function being applied, plugged on the $\delta$ area.

\begin{figure}[tb]
\centering
$\infer[(app)]
  { R;\Gamma \vdash \refssl{}{\multv{V}}{M\ N} : (\alpha, e=e_1 \cup e_2 \cup e_3) }
  { R;\Gamma \vdash M : (A \sur{\to}{e_1} \alpha, e_2)
    \qquad R;\Gamma \vdash N : (A,e_3)
    \qquad R;\Gamma \vdash V \in \multv{E}_i : (R(r_i),\emptyset)}$
\centering
\includegraphics{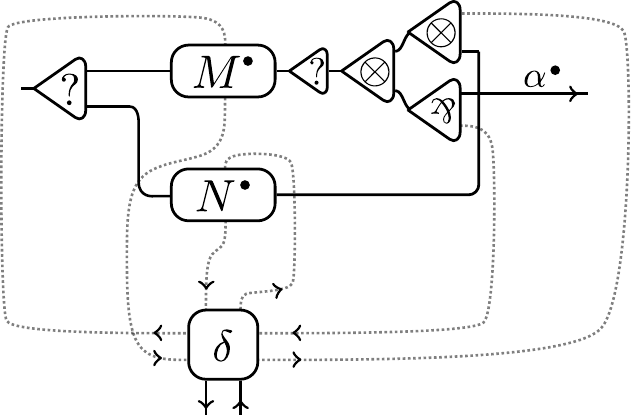}
\caption{Translation of $M\ N$}\label{figure:app}
\end{figure}

\subsection{Full translation of {\lthis} into proof nets}
We presented some interesting cases of the translation of {\lamadio} to proof
nets to provide the reader with some intuition. However the complete translation
is rather operating on the
intermediate language {\lthis}. A translation and a simulation theorem between
{\lamadio} and {\lthis} are given in~\cite{HamdaouiValiron2018}, completing the
picture.

\begin{center}
\begin{longtable}{|c|c|}
\hline
Typing derivation & Translation \\
\hline
$\infer[(var)]{\Gamma,x : A \vdash x : (A,\emptyset)}{}$
& \includegraphics{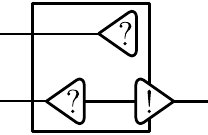} \\
\hline

$\infer[(unit)]{\Gamma \vdash \ast : (\texttt{Unit},\emptyset)}{}$
& \includegraphics{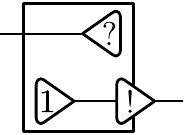} \\
\hline

$\infer[(lam)]
{ \Gamma \vdash \lambda{x}.M : (A \sur{\to}{e} \alpha, \emptyset) }
{ x:A,\Gamma \vdash M : (\alpha,e) }$
& \includegraphics{figures/translation/pdfs/translation-lam.pdf} \\
\hline

$\infer[(app)]
{ \Gamma \vdash \refssl{}{\multv{V}}{M\ N} : (\alpha, e=e_1 \cup e_2 \cup e_3) }
{ \Gamma \vdash M : (A \sur{\to}{e_1} \alpha, e_2)
  \qquad \Gamma \vdash N : (A,e_3)
  \qquad \Gamma \vdash V \in \multv{E}_i : (B_i,\emptyset)}$
& \includegraphics{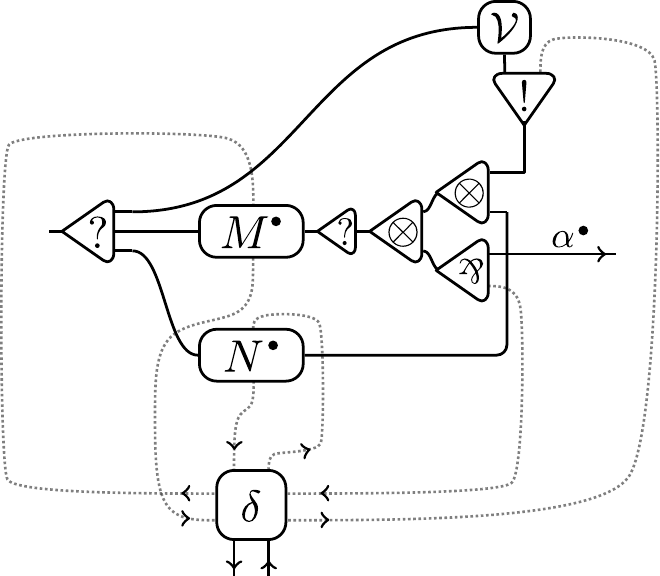} \\
\hline

$\infer[(get)]
{ \Gamma \vdash \get{r} : (A,\{ r \}) }
{ \vdash \Gamma }$
& \includegraphics{figures/translation/pdfs/translation-get.pdf} \\
\hline

% $\infer[(set)]
% { R;\Gamma \vdash set(r,V) : (\text{Unit},\{ r \}) }
% { R \vdash r:\text{Reg}_r\ A \qquad R;\Gamma \vdash V : (A,\emptyset) }$
% & \tikzvcenter{
%   \tikztarget(set){
%     \begin{proofnet}[baseline=-.5ex]
%       \init[grow=left]{
%         \boxpal[name=ppal]{p}{
%           \nary[wire=arrowed,shape=rectangle,cell={net={15pt}{20pt}},name=V,fix=both]{V^\bullet}{\wire[inbox={}]{\coord(vgamma)}}
%         }
%       }
%       \init[grow=down]([yshift=6ex,xshift=7.5ex]V.east){
%         \zeroary[name=myweak1]{\wn}
%       }
%       \init[grow=left]([xshift=15ex]V.south){
%         \boxpal{pone}{
%           \zeroary[name=one]{1}
%         }
%       }
%       % \init[grow=right]{
%       %   \flipped[small,name=myweak1]{\wn}{}{\point(pend1)}
%       % }
%       % \init[grow=left]([yshift=20pt,xshift=21pt]myweak1){
%       %   \boxpal{p}{
%       %     \unary[flip]{\wn}{
%       %       \wire[inbox={}]{\wire[inbox={}]{\point(pend)}}
%       %     }
%       %   }
%       % }
%       % \addtobox{p}{myweak1}
%       % \draw let \p1=(pend), \p2=(pend1) in (\x1,\y2) -- (pend1);
%       % \draw[reverse arrowed] (pend) -- ++(0,25pt);
%       % \drawbox{p}
%       % \init[grow=up]([yshift=-30pt]myweak1){
%       %   \wire{\zeroary{\oc}}
%       % }
%       \drawbox{p}
%       \drawbox{pone}
%       \wirefromto{[xshift=5pt]ppal.pal}{[xshift=5pt]ppal.pal}{[yshift=-15pt,xshift=5pt]ppal.pal}{[yshift=-15pt,xshift=5pt]ppal.pal}
%       % \draw[rounded corners=3pt, arrowed] ([xshift=5pt]ppal.pal) |- ++(0,-15pt);
%     \end{proofnet}
%   }
% } \\
% \hline

$\infer[(par)]
{ R; \Gamma \vdash P_1 \parallel P_2 : \mathbf{B} }
{i=1,2 \qquad R;\Gamma \vdash P_i : \alpha_i }$
& \includegraphics{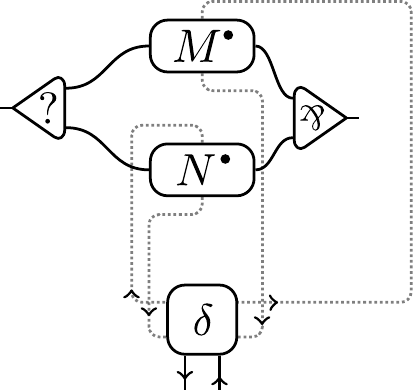} \\
\hline

\infer[(subst)]
{ R; \Gamma \vdash \varss{\sigma}{M} : (\alpha, e) }
{ R; \Gamma, x : A \vdash M : (\alpha,e) \qquad R; \Gamma \vdash V : (A,\emptyset) }
& \includegraphics{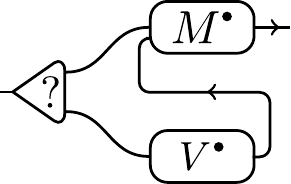} \\
\hline

\infer[(subst\text{-}r\downarrow)]
{ R; \Gamma \vdash \substg{(r_i)}{(\mathcal{V}_i)}{M} : (\alpha, e) \qquad }
{ R; \Gamma \vdash M : (\alpha,e) \qquad R; \Gamma \vdash \mathcal{V}_i : (A_i,\emptyset) }
& \includegraphics{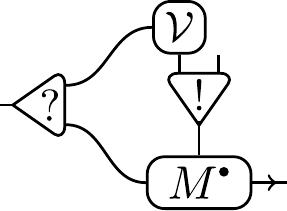} \\
\hline

\infer[(subst\text{-}r\uparrow)]
{ R; \Gamma \vdash \substup{(r_i)}{(\mathcal{V}_i)}{M} : (\alpha, e) \qquad }
{ R; \Gamma \vdash M : (\alpha,e) \qquad R; \Gamma \vdash \mathcal{V}_i : (A_i,\emptyset) }
& \includegraphics{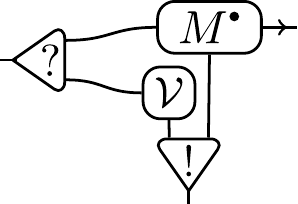} \\
\hline
\end{longtable}
\end{center}

We proceed to give the properties satisfied by the translation.

\section{Properties of the translation}
\label{section:properties}

\subsection{Simulation}
The first result is the simulation theorem. The formulation precises
that a deterministic step in {\lamadio} is mapped to a deterministic reduction in nets, and
only the reduction of $\get{r}$ produces a non-deterministic sum.

\begin{theorem}{Simulation}\label{theorem-nets-simulation}\\ Let $\vdash M :
  (\alpha,e)$ be a closed well-typed term of \lamadio. Then
  \begin{itemize}
    \item If $M \to N$  by $(\beta_v)$ or $(\texttt{set})$ then
      $M^\bullet \to^\ast N^\bullet$
    \item If $M \to N_i$ by $(\texttt{get})$ then
      $M^\bullet \to^\ast \net R + \sum_i {N_i}^\bullet$
  \end{itemize}
\end{theorem}

The presence of the additional term $\net R$ is linked to the
local nature of non-determinism in proof nets. When facing a get
reduction, the proof net can either select one of the available {\setw}, or drop
them all and wait for an hypothetical future one, which corresponds to this
$\net R$ net. This is better understood when stating the simulation theorem on
{\lthis}:

\begin{theorem}{Simulation for {\lthis}}\label{theorem-nets-simulation-lthis}\\
  Let $\vdash M : (\alpha,e)$ be a closed well-typed term of \lthis
  . If $M \to^\ast N$, then $M^\bullet \to^\ast N^\bullet$.
\end{theorem}

Here, the term $\net R$ is totally internalized in {\lthis} and we get a clean
simulation theorem.  To derive this result, we start by defining a notion of
typed context and a translation of contexts to nets.

\begin{definition}{Hole typing rule}\label{definition-nets-simulation-typehole}\\
\[
  \infer[(hole)]
  {R,[.] : (\alpha,e) \vdash [.] : (\alpha,e)}
  {R \vdash (\alpha,e)}
\]
\end{definition}

\begin{center}
  \begin{tabular}{|c|c|}

  \end{tabular}
\end{center}

$[.]$ can be seen just as a special kind of variable in the typing derivation
that can be substituted by something else than a value. We use the usual typing
rules to build derivations of typed contexts. A context can then be translated
to a net, with an interface determined by $\alpha$ and $e$, labelling
free ports.  This is what we call net contexts. The substitution of nets
consists in plugging the translation of a term with a matching type in this interface.

\begin{definition}{Hole translation}\label{definition-nets-simulation-holetrans}\\
  We define the translation of a typed hole $R \vdash [.] : (\alpha,e)$
as

\begin{center}
  \begin{tabular}{DD}
  \[
    \infer[(hole)]
    {R,[.] : (\alpha,e) \vdash [.] : (\alpha,e)}
    {R \vdash (\alpha,e)}
  \]
   &
    \includegraphics{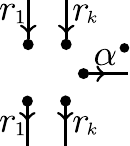} \\
  \end{tabular}
\end{center}

We can then carry on and use the usual term translation to build the
translation of a context $(C[E])^\bullet$, which is a net of the form

\begin{center}
  \includegraphics{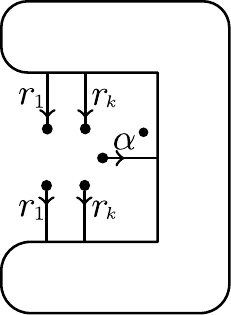}
\end{center}

The substitution of $M^\bullet$, or of any net with a compatible interface for
that matter, in $(C[E])^\bullet$ is defined by just connecting the free wires of
the substituted net to the corresponding free wire of the context hole :

\begin{center}
  \includegraphics{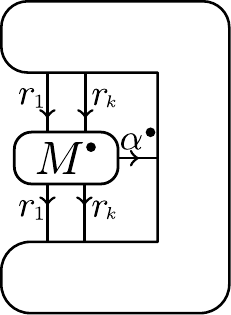}
\end{center}

\end{definition}

The fundamental property of net contexts and net substitutions is that the
substitution commutes with the translation, in the following sense :

\begin{property}{Nets substitution}\label{property-nets-substitution}\\
  Let $ \vdash [.] : (\alpha,e)$ be a typed hole, $[.] : (\alpha,e) \vdash C[E]
  : (\beta,e')$ a typed context, and $ \vdash M : (\alpha,e)$ a term. Then :
\[
  (C[E[M]])^\bullet = (C[E])^\bullet[M^\bullet]
\]
\end{property}

% The proof is an easy induction: the base case when $C[E] = [.]$ is immediate,
% and by definition of the context translation, so is the inductive case.
%
The very definition of net reduction immediately entails that if $M^\bullet
\to^\ast N^\bullet$ then $(C[E])^\bullet[M^\bullet] \to^\ast
(C[E])^\bullet[N^\bullet]$. Together with
Property~\ref{property-nets-substitution}, this ensures that we can focus on the case
where $C = E = [.]$, as the general case follows seamlessly.

From here, we check that each reduction rule of \lthis\ can be simulated on the
net side, relying on the definition of the reduction and the behavior of routing areas.

% \paragraph{Reference substitutions}
%
% We state and prove a few lemma about the reduction of the translation of reference
% substitutions.
%
% [TODO]

\paragraph{Variable substitutions reductions}

Let us show the simulation for \rrulet{subst}{} rules, involving mainly
duplication. Thanks to the previous theorem, we can assume that contexts are
empty without loss of generality. We consider only closed terms. Indeed, since
reduction contexts $S,C,E$ do not bind variables, all the terms appearing in the
premise of a rule are thus closed terms, and we can
omit their context. For each rule, we write the translation of the premise
followed by its reduction in nets, which matches the conclusion.

The fundamental rule is the variable one.

\begin{center}
  \begin{longtable}{lDED}
    \rrulet{subst}{var} ($\sigma(x)$ undefined)
    & \includegraphics{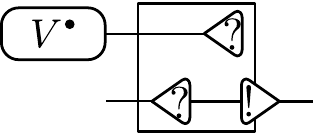}
    & $\to^\ast$
    & \includegraphics{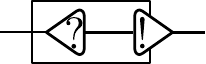} \\

    \rrulet{subst}{var} ($\sigma(x)$ defined)
    & \includegraphics{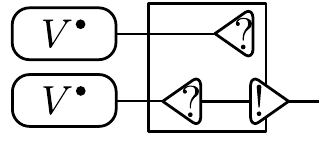}
    & $\to^\ast$
    & \includegraphics{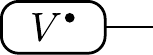} \\
  \end{longtable}
\end{center}

\vspace{0.5cm}
When reaching a $\get{r}$ or a $\ast$, the substitution simply vanishes.
\begin{center}
 \begin{longtable}{lDED}
    \rrulet{subst}{unit}
    & \includegraphics{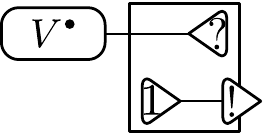}
    & $\to^\ast$
    & \includegraphics{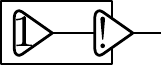} \\
  \end{longtable}
\end{center}

\vspace{0.5cm}
The \rrulet{subst}{app}, \rrulet{subst}{subst-r} and \rrulet{subst}{subst-r'}
perform a duplication and propagate the substitution inside reference
substitutions.

\begin{center}
  \begin{longtable}{lDED}
    \rrulet{subst}{app}
    & \multicolumn{3}{D}{
      \includegraphics{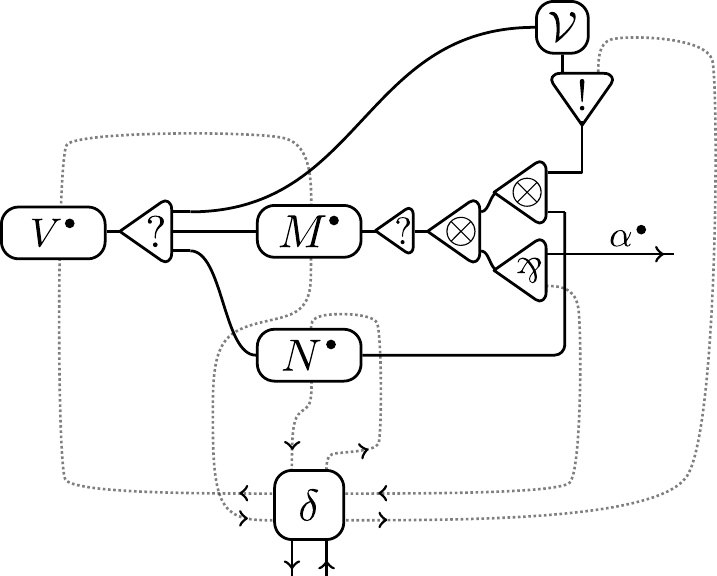}
    } \\

    $\to^\ast$
    & \multicolumn{3}{D}{
      \includegraphics{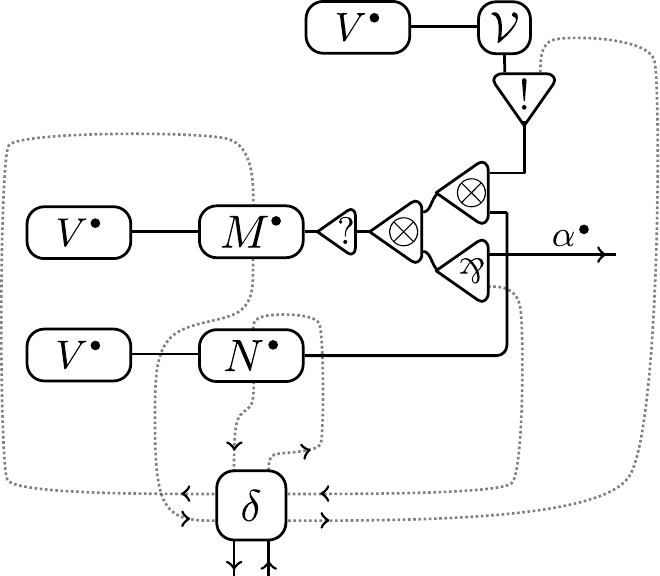}
    } \\

    \rrulet{subst}{subst-r}
    & \includegraphics{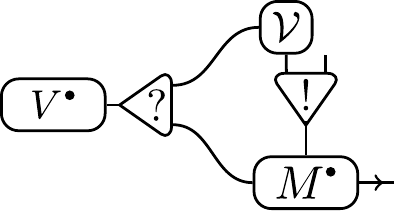}
    & $\to^\ast$
    & \includegraphics{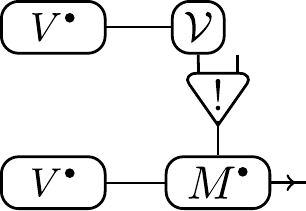} \\

    \rrulet{subst}{subst-r'}
    & \includegraphics{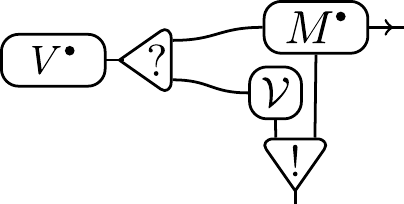}
    & $\to^\ast$
    & \includegraphics{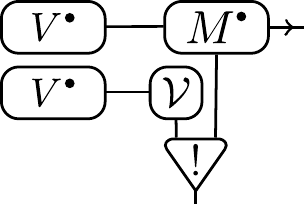} \\
  \end{longtable}
\end{center}

\vspace{1cm}
The \rrule{subst}{\parallel} just duplicate the variable substitution to the two
threads
\begin{center}
  \begin{longtable}{lDED}
    \rrule{subst}{\parallel}
    & \includegraphics{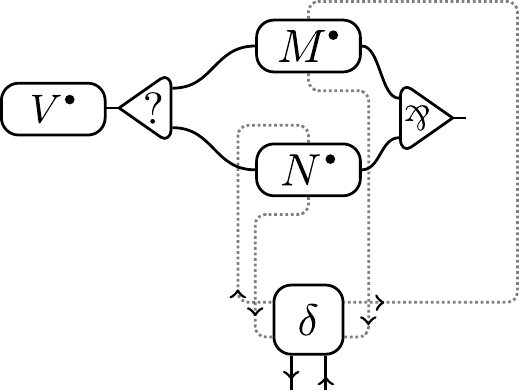}
    & $\to^\ast$
    & \includegraphics{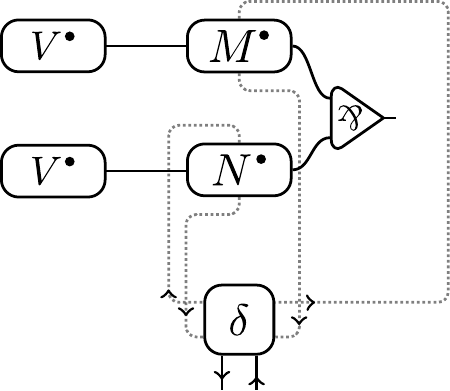} \\
  \end{longtable}
\end{center}

\vspace{0.5cm}
Finally, the \rrulet{subst}{merge} distributes the outter substitution to both
the term and the inner substitution.
\begin{center}
  \begin{longtable}{lDED}
    \rrulet{subst}{merge}
    & \includegraphics{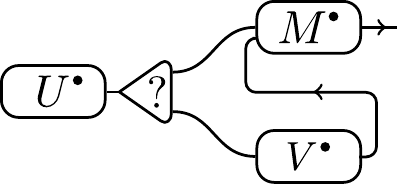}
    & $\to^\ast$
    & \includegraphics{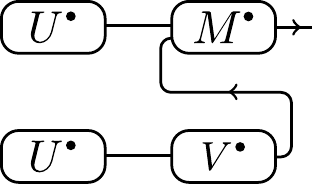} \\
  \end{longtable}
\end{center}

\paragraph{Downward reference substitutions reductions}

The propagation of references substitutions relies on the behavior of routing
area, and especially Property~\ref{property-translation-transit}. The
fundamental case is the non deterministic reduction happening when reducing a
$\get{r}$ whose redex is
\begin{center}
  \begin{longtable}{lD}
    \rrulet{subst-r}{get}
    & \includegraphics{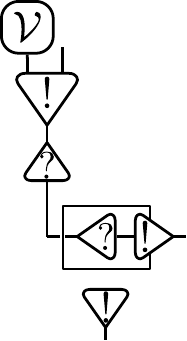} \\
  \end{longtable}
\end{center}

Then for each $V$ in the image of $\multv{V}$, there will be exactly one summand
of the following form

\begin{center}
  \begin{longtable}{DED}
    \includegraphics{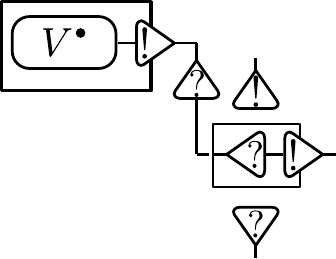}
    & $\to^\ast$
    & \includegraphics{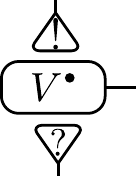} \\
  \end{longtable}
\end{center}

The remaining term does indeed reduce to the translation of $\get{r}$:
\begin{center}
  \begin{longtable}{DED}
    \includegraphics{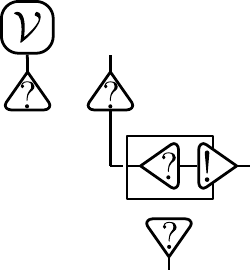}
    & $\to^\ast$
    & \includegraphics{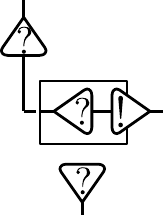} \\
  \end{longtable}
\end{center}

\vspace{0.5cm}
\rrulet{subst-r}{val}, \rrule{subst-r}{\parallel} and \rrulet{subst-r}{app} are
just direct application of the Property~\ref{property-translation-transit}.

\begin{center}
  \begin{longtable}{lDED}
    \rrulet{subst-r}{val}
    & \includegraphics{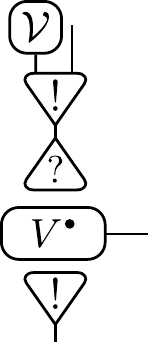}
    & $\to^\ast$
    & \includegraphics{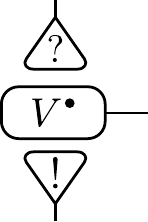} \\

    \rrule{subst-r}{\parallel}
    & \includegraphics{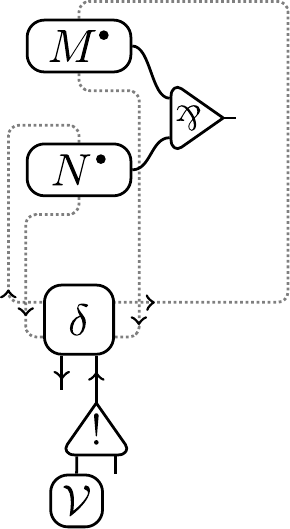}
    & $\to^\ast$
    & \includegraphics{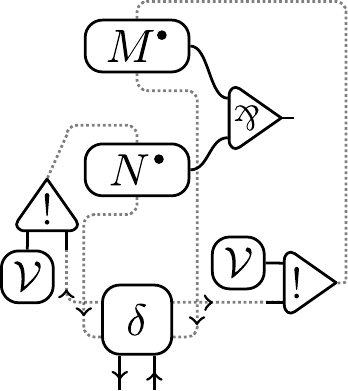} \\

    \rrulet{subst-r}{app}
    & \multicolumn{3}{D}{
      \includegraphics{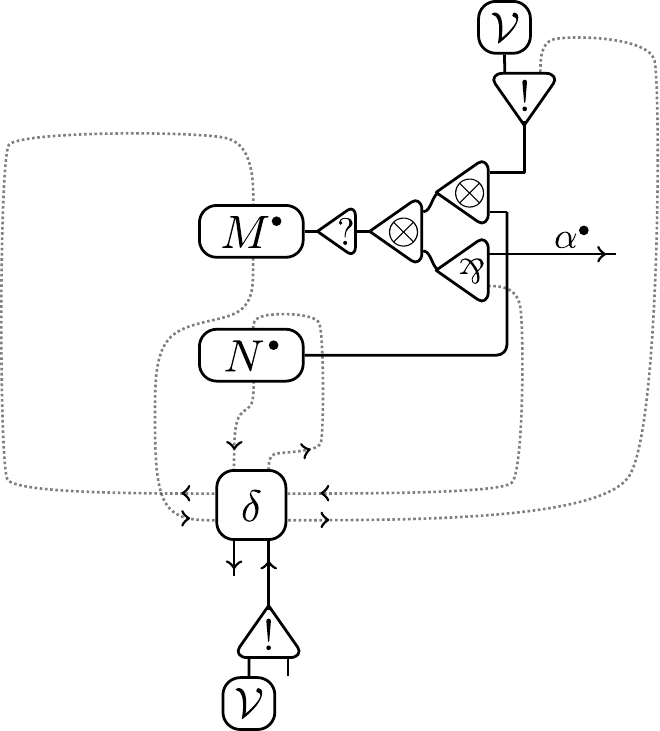}
    } \\
    $\to^\ast$
    & \multicolumn{3}{D}{
      \includegraphics{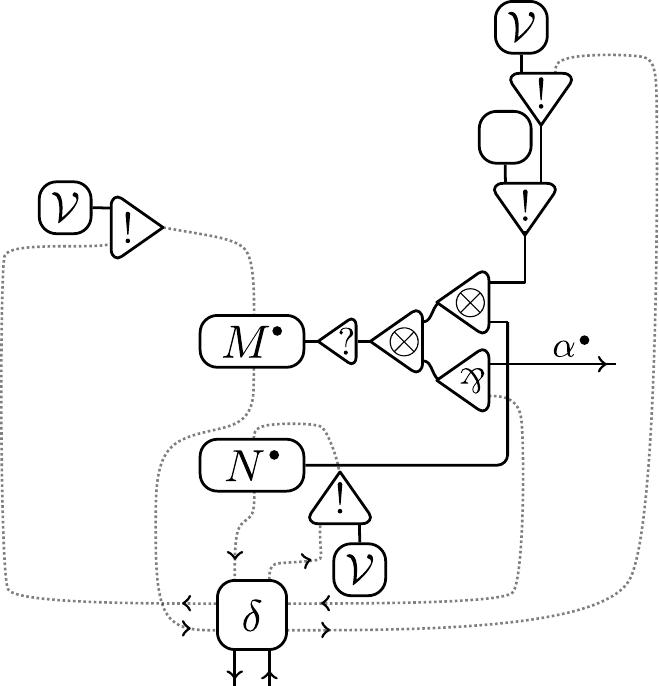}
    } \\
  \end{longtable}
\end{center}

\vspace{0.5cm}
\rrulet{subst-r}{merge} and \rrulet{subst-r}{subst-r'} amount to nothing in nets, as the translation
already identifies the redex and the reduct of these rules.

\begin{center}
  \begin{longtable}{lDlD}
    \rrulet{subst-r}{subst-r'}
    & \includegraphics{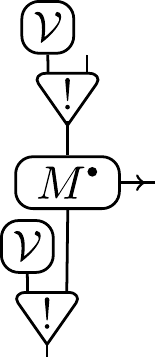}
    & \rrulet{subst-r}{merge}
    & \includegraphics{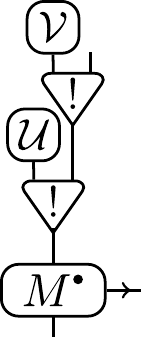} \\
  \end{longtable}
\end{center}

\paragraph{Upward reference substitutions reduction}
As for downward substitutions, the main ingredient is the Transit lemma applies
to our specific routing area $\delta$ and $\gamma$.

\begin{center}
  \begin{longtable}{lDED}
    \rrule{subst-r'}{\parallel}
    & \includegraphics{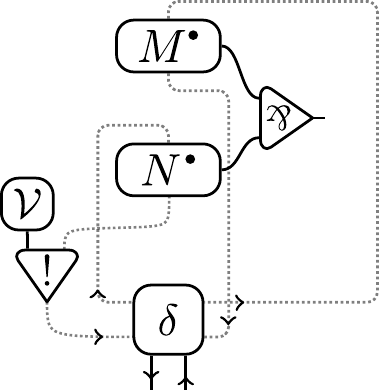}
    & $\to^\ast$
    & \includegraphics{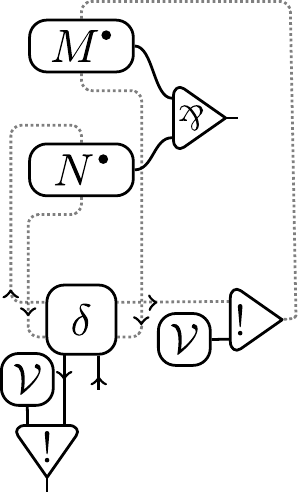} \\

    \rrulet{subst-r'}{lapp}
    & \includegraphics{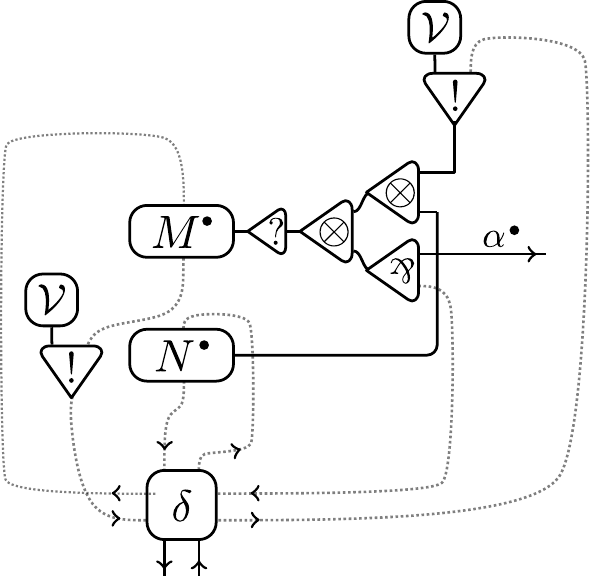}
    & $\to^\ast$
    & \includegraphics{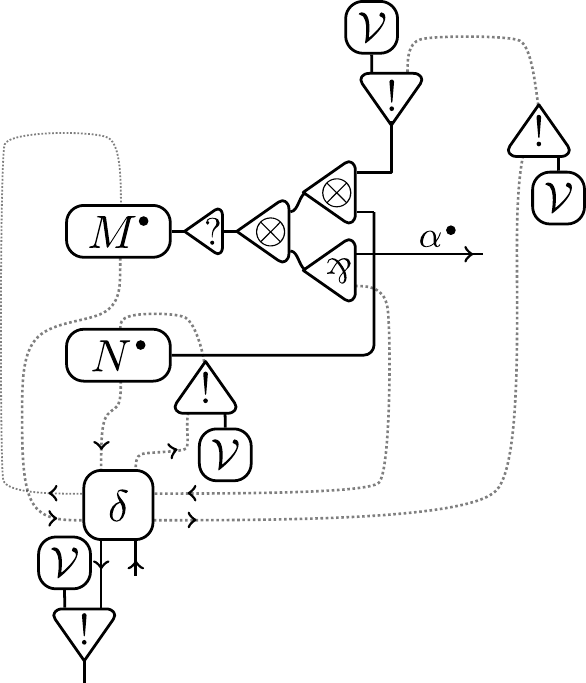} \\

    \rrulet{subst-r'}{rapp}
    & \includegraphics{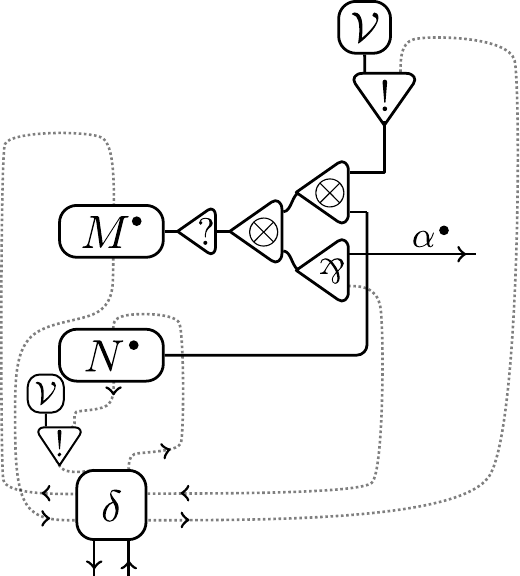}
    & $\to^\ast$
    & \includegraphics{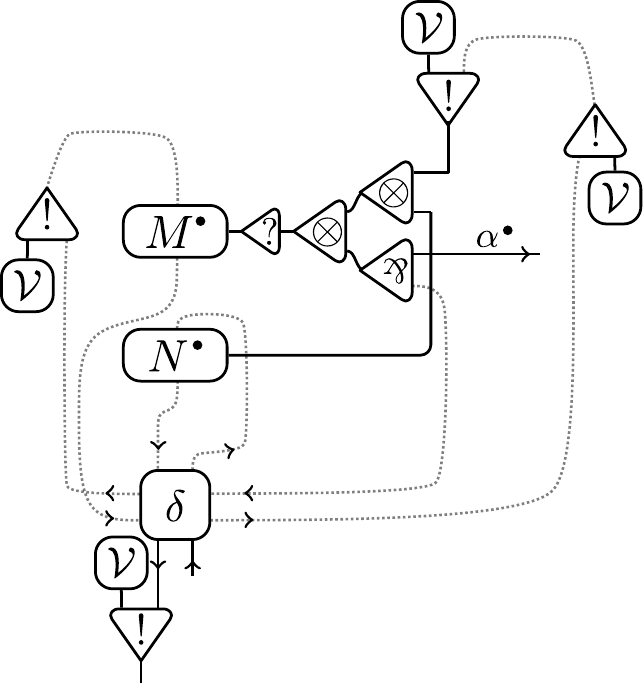} \\
  \end{longtable}
\end{center}
\newcommand{\labelsubst}[2]{|{#1}|_{#2}}
\newcommand{\lthisplus}{$\lamlang{cES+}$}

\subsection{Termination}
The second result states that the translation of a term is strongly
normalizing:

\begin{theorem}{Termination}\label{theorem-termination}\\
  The translation of a well typed term of {\lamadio} terminates.
\end{theorem}

The theorem that we will actually prove is rather the following:
\begin{theorem}{Termination of \lthis}\label{theorem-termination-lthis}\\
  The translation of the normal form of a well-typed term of {\lthis} terminates.
\end{theorem}

While {\lthis} is closer to nets that {\lamadio}, unfortunately a normal form of
{\lthis} is not translated to a normal form in proof nets. The corresponding net
can still perform some reductions. But we will see that these are unessential
and limited: in a few steps, a normal form is reached in nets.

We proceed in two stages: first, we extend {\lthis} to a language {\lthis +}
that is able to do just a little more reductions than {\lthis}. The extension of
the translation and the simulation theorem for {\lthis +} are straightforward.
Then, we show that {\lthis +} also terminates, give an explicit grammar for its
normal forms, and show that the translation of these normal forms are
strongly normalizing proof nets.

\paragraph{\lthis +}

We add a labelled reference substitution variable substitution to {\lthis}
: $\refssdbar{}{V}{M} \mid \varsbar{x}{N}{M}$. The first one,
$\refssdbar{}{V}{M}$, correspond to a downward reference substitution where
$\multv{V}$ has no free variables. The second one corresponds to a variable
substitution of a term that is not a value. Both can appear during the reduction
in proof nets but are not accounted for in {\lthis}. $\varsbar{x}{N}{M}$ can't
be reduced, either in $N$ or by \rrulet{subst}{} rules (whether $N$ is a value or
not). The labelled reference forbids any reduction under it, except for the
labelled reference substitutions.  We extend the reduction rules
\rrulet{subst-r}{} to the labelled substitution $\refssdbar{}{V}{}$. They can be
performed anywhere (included under a variable substitution) except under
abstraction, in a context defined by the following grammar: $J ::= [.] \mid J\ M
\mid M\ J \mid \refssd{}{V}{J} \mid \varss{\sigma}{J} \mid \varsbar{x}{N}{J}$.

We also add a $\beta$-rule to fire such labelled substitution :
\[
  (\beta_{V}') \refssl{}{V}{(\lambda x . M)\ N} \to
  \varsbar{x}{N}{\refssdbar{}{V}{M}}
\]

This corresponds to the additional reduction nets can do. Indeed a $\beta$ redex
can always be fired in nets even if the argument is not a value. But then, the
corresponding term do need to be reduced before any duplication, erasure or
substitution. Let us now show termination and describe the normal forms of {\lthis+}:

\begin{mdef}{$F$-normal forms}\label{lemma-progress-auxplus}\\
  Let $M$ be a normal form of the form of {\lthis}. It belongs to the grammar
  $M_\text{norm}$ (cf~\cite{HamdaouiValiron2018}). Then $M$ reduces to a sum
  of terms in {\lthis +} that belongs to the
 following grammar of $F$-normal forms :
\begin{tabred}
  F_\text{norm} & ::= & \get{r} \mid \refssl{}{V}{F_\text{norm}\ V} \mid
  \refssl{}{V}{F_\text{norm}\ F_\text{norm}} \mid \varsbar{x}{F_\text{norm}}{M} \\
\end{tabred}
where the $M$ is in \lthis (contains no labelled substitution).
\end{mdef}

\begin{proof}
  By induction on the structure of $M$, with the additional hypothesis that only
  the additional rules are performed (no upward substitution) :
  \begin{itemize}
    \item $M = \get{r}$ : ok 
    \item $M = \refssl{}{V}{M_\text{norm}\ V}$ : by induction, $M_\text{norm}$
      reduces to a sum of $F_\text{norm}$. Take one such summand $N$, then $M \to^\ast
      \refssl{}{V}{N_\text{norm}\ V}$ (because no rule \rrule{subst-r'}{\top} was
      used)
    \item $M = \refssl{}{V}{M_\text{norm}\ M'_\text{norm}}$ : we proceed as the
      previous case
    \item $M = \refssl{}{V}{V\ M_\text{norm}}$ : by induction, $M_\text{norm}$
      reduces to some $N$ in $F_\text{norm}$ grammar, so $M \to^\ast
      \refssl{}{V}{V\ N}$. Then, by inversion of typing rules, $V$ is of the
      form $\lambda x . P$ and we can apply the new $\beta$-reduction to get
      $\varsbar{x}{N}{\refssdbar{}{V}{P}}$. By pushing down $\refssdbar{}{V}{}$ in
      $P$, we can reduce it to a sum of $P_i$s where each $P_i$ do not contain
      labelled substitutions. Then $M \to^\ast \sum_i \vars{x}{N}{P_i}$
  \end{itemize}
\end{proof}

\begin{lemma}
$F$-normal forms are normal.
\end{lemma}

\begin{proof}
  By induction :
  \begin{itemize}
    \item $M = \get{r}$ : ok
    \item $M = \refssl{}{V}{F_\text{norm}\ V}$ : by induction, $F_\text{norm}$
      is normal, and values are not $F$-normal form, so the application can't create any
      $\beta$-redex.
    \item $M = \refssl{}{V}{M_\text{norm}\ M'_\text{norm}}$ : same as the
      previous case
    \item $M = \varsbar{x}{F_\text{normal}}{M}$ : the explicit substitution
      prevents any reduction except \rrulet{subst-r}{} ones for labelled
      reference substitution, but by definition, $M$ does not contain any.
  \end{itemize}
\end{proof}

\paragraph{$F$-normal forms and proof nets}
We will see in the following that the translation of a $F$-normal form has not
many possible reductions left.  However, a few steps may remain to eventually
reach a normal form. The first step is to collect and merge all the routing
areas that are created and connected during the translation. Doing so, we
separate the net between a part that closely follows the translated term
structure, and a big routing area which connects various subterms to enable
communication through references between them. Once this is done, a few starving
{\getw s} may interact with the routing area, but nothing more. The following
definition give the shape obtained after the merging of routing:

\begin{definition}{Separability}\label{definition-nets-termination-separability}\\
  Let $\mathcal{R}$ be the translation of a {\lthis +}
  term, then $\mathcal{R}$ is say to be \emph{separable} if
  it can be reduced to the following form :

  \begin{center}
    \includegraphics{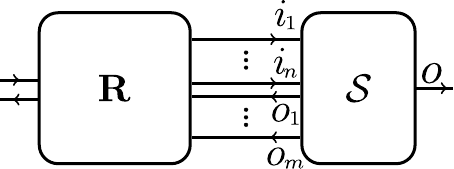}
  \end{center}

  % \begin{center}
  %   \tikzvcenter{
  %     \begin{proofnet}
  %       % Communication box
  %       \node[net={40pt}{40pt}] (box) {\mathbf{R}};
  %       % Remaining part of the net
  %       \init[grow=left]([xshift=15ex]box.east){
  %         \nary[shape=rectangle,cell={net={40pt}{30pt}},name=sigma]{\mathcal{S}}{
  %           \point(sigmai1),
  %           \point(sigmao1),
  %           \point(sigmain),
  %           \point(sigmaon)
  %         }
  %       }
  %       % Box coordinates
  %       \coordinate (i1) at ([yshift=-13pt]box.east);
  %       \coordinate (o1) at ([yshift=-8pt]box.east);
  %       \coordinate (in) at ([yshift=8pt]box.east);
  %       \coordinate (on) at ([yshift=13pt]box.east);
  %       \coordinate (i0) at ([yshift=-3pt]box.west);
  %       \coordinate (o0) at ([yshift=2pt]box.west);
  %       % Sigma coordinates
  %       \coordinate (sigmai1) at ([yshift=-13pt]sigma.north);
  %       \coordinate (sigmao1) at ([yshift=-8pt]sigma.north);
  %       \coordinate (sigmain) at ([yshift=8pt]sigma.north);
  %       \coordinate (sigmaon) at ([yshift=13pt]sigma.north);
  %       % Connect box and sigma
  %       \init[grow=up]{
  %         \wirefromto{sigmai1}{sigmai1}{i1}{i1}
  %         \wirefromto{sigmain}{sigmain}{in}{in}
  %         \wirefromto{sigmao1}{sigmao1}{o1}{o1}
  %         \wirefromto{sigmaon}{sigmaon}{on}{on}
  %         \wirefromto{[xshift=-10pt]i0}{[xshift=-10pt]i0}{i0}{i0}
  %         \wirefromto{[xshift=-10pt]o0}{[xshift=-10pt]o0}{o0}{o0}
  %       }
  %       %
  %       \pdots[2pt]([xshift=10pt]o1)([xshift=10pt]in)
  %     \end{proofnet}
  %   }
  % \end{center}
  where $\mathbf{R}$ is a routing area and $\mathcal{S}$ a net with free wires labelled by
  $i_1,\ldots,i_n,o_1,\ldots,o_m,O$, satisfying :
  \begin{description}
    \item[(a)] There is no redex in $\mathcal{S}$
    \item[(b)] $i_1,\ldots,i_n$ are either connected to the auxiliary port of a
      $\otimes$ cell, to the auxiliary port of a cocontraction or to the principal port of a dereliction
    \item[(c)] $o_1,\ldots,o_m,O$ are either connected to a
      $\parr$ cell, or to the principal door of an open box
  \end{description}
\end{definition}

The translation of term with at least one free variable is separable.  The
reason is that constructors such as application, substitution, parallel
composition, etc. preserve separability. Moreover, the translation of values
with free variables are obviously separable. The weakenings, which correspond to
free variables, materialize as auxiliary doors of exponential boxes, thus
blocking further reduction.

\begin{lemma}{Open terms separability}\label{lemma-nets-termination-otsep}\\
  Let $\Gamma,x:A \vdash M : (\alpha,e)$ a {\lthis} well-typed term, then its
  translation $M^\bullet$ is separable.
\end{lemma}

\begin{proof}\ref{lemma-nets-termination-otsep}\\
  We proceed by induction on terms.
\begin{description}
  \item[Values] As stated above, we can first observe that the translation of $\ast$, of a
    variable $x$ and an abstraction all satisfy the separability conditions.
    Indeed, they are composed of a box with at least one auxiliary door, and the
    inside of the box is a normal form (by induction for abstraction and
    trivially for others). As they are pure terms, $m = n = 0$.
  \item[Get] It is almost the same as values, except that the output $o$
    corresponding to the reference of $\get{r}$ is connected to a dereliction, which is
    allowed in \textbf{(b)}.
  \item[Application] By induction, $M^\bullet$ and $N^\bullet$ are
    separable, thus can be decomposed in the following way :

    \begin{center}
      \includegraphics{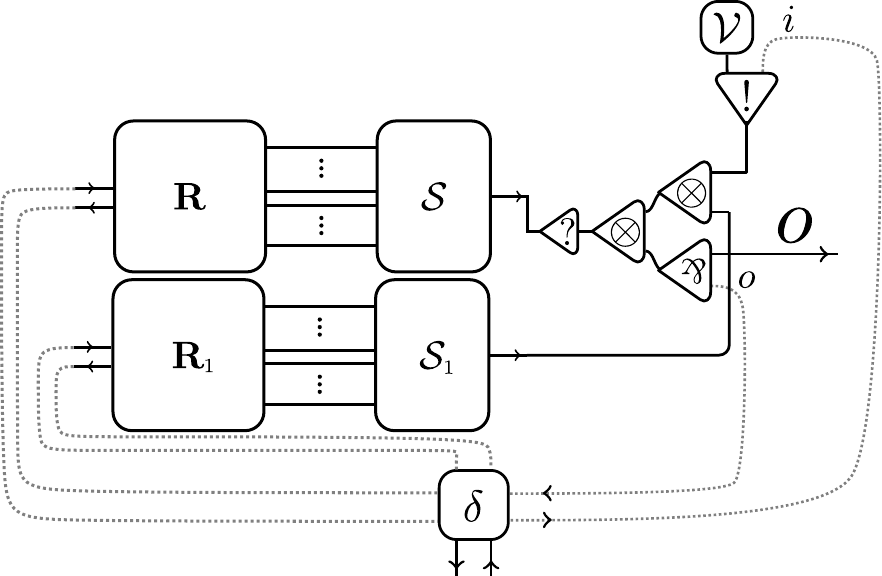}
    \end{center}

    We can merge the 3 routing areas $\routarea{R}, \routarea{R}_1$ and
    $\delta$. All the inputs or outputs previously connected to one of the small
    areas immediately satisfy the conditions \textbf{(b,c)} by IH. The two
    remaining wires $i$ and $o$ are respectively connected to the auxiliary port
    of a par and the auxiliary port of a cocontraction, thus satisfy
    \textbf{(b,c)}.  $\mathcal{S}$ and $\mathcal{S}_1$ are normal by IH, and the
    translation of $\multv{V}$ is easily seen as normal. $O$ is the conclusion of
    a par, hence satisfies \textbf{(b,c)}. It only remains to see that the whole
    net excepted the routing areas is normal, but the only redex that could appear
    is the connection of a dereliction to $M^\bullet$. By \textbf{(c)},
    the output of $M^\bullet$ is either the conclusion of an open box or a par which do not form a
    redex.  The condition \textbf{(a)} is verified.
  \item[Parallel] Similar to application
  \item[Variable substitution] As for application, we apply the IH on
    $M^\bullet$ and $V^\bullet$, merge the routing areas, and just check that the
    connection of $V^\bullet$ to $M^\bullet$ can not create new redexes using
    \textbf{(b)} on the output wire of $V^\bullet$.
  \item[Reference substitutions] Again, the same technique is applied.
\end{description}

\end{proof}

Finally, the main result we rely on as explained above is the following one :

\begin{lemma}{Normal form separability}\label{lemma-nets-termination-nfsep}\\
  The translation of a $F$-normal form is separable
\end{lemma}

It is proved as Lemma~\ref{lemma-nets-termination-otsep}, by induction on the
syntax of $F$-normal forms. We can eventually prove from this last lemma:

\begin{lemma}{Termination of {\lthis +}}\label{lemma-termination-lthisp}\\
  The translation of a $F$-normal form is strongly normalizing.
\end{lemma}

From which we deduce:

\begin{corollary}{Strong normalization}\label{corollary-termination-sn}\\
  \begin{enumerate}
    \item The translation of a closed well-typed term of {\lthis +} is strongly normalizing.
    \item The translation of a closed well-typed term of {\lthis} is strongly
      normalizing (Theorem~\ref{theorem-termination-lthis}).
  \end{enumerate}
\end{corollary}

\begin{proof}{Lemma~\ref{lemma-termination-lthisp}}\\
  We apply \ref{lemma-nets-termination-nfsep} to reduce the translation of a
  $F$-normal form $F$ to a routing area $\routarea R$ connected to $\mathcal{S}$
  satisfying separability conditions. $\mathcal{S}$ and $\routarea R$ are
  normal. The potential redexes must involve a wire at the interface of
  $\routarea{R}$ and $\mathcal{S}$. The inputs of $\routarea{R}$ are connected
  to $\mathcal{S}$, either to a par or to the conclusion of an open box and thus
  can't form any redex with (co)weakenings and (co)contractions of $\routarea
  R$.

  The outputs of $\routarea R$ are either connected to the auxiliary port of a
  tensor, the auxiliary port of a cocontraciton or to a dereliction. Only the
  latter may form a new redex. If the output of $\mathbf{R}$ is a coweakening,
  then everything reduces to $\mathbf{0}$. If it is just a wire, then the
  dereliction is connected through this wire to an input of $\routarea R$ which
  again can't be part of any redex. The only remaining case is when the output
  of $\routarea R$ is a cocontraction tree.  Then we can perform the
  non-deterministic $\red{ba}$ reductions, and we go back to exactly the two
  previous cases as the dereliction is finally connected to a leaf of the tree
  of an input.

  Hence, after we performed finitely many $\red{ba}$ reductions,
  we finally get a normal form, which is either $\mathbf{0}$, or a sum of the
  previous net $\mathbf{S}$ connected to simpler routing area $\routarea{R}_i$.
\end{proof}

Simulation and termination ensure that if a term $M$ reduces to $T = V_1 \parallel
\ldots \parallel V_n$, then:
\begin{itemize}
  \item By simulation, the translation $M^\bullet$ can be reduced to a non
    deterministic sum in which one summand will be $T^\bullet$
  \item By termination, this does not depend on the path of reduction we chose:
    any reduction will converge to a normal form whose summands contain the
    net $T^\bullet$.
\end{itemize}

The normal form of $M^\bullet$ must contain the summand $T$. However, this does
not tell us anything about what are the other summands of $M^\bullet$. Let us
pretend the support for integers in \lamadio. Assume that a term $M$ have $2$ as
only normal form. Because of this additional term $\net R$ appearing in the
simulation theorem, our results do not prevent $M^\bullet$ to reduce to
$1^\bullet + 2^\bullet + 3^\bullet$. In such a case, one may argue that the
proof nets do not reflect faithfully the language as they may have a lot more
possible outcomes.

\subsection{Adequacy}

The adequacy theorem states that the summands of the normal form of
$M^\bullet$ are either the translation of a normal form $T = V_1 \parallel
\ldots \parallel V_n$ that is a reduct of the original $M$, or garbage, that is
a non correct net that corresponds to execution paths which deadlocked. We can
recognize this garbage, thus eliminate it: with this additional operation, the
summands of the normal form of $M^\bullet$ coincide with the values that are
reachable by $M$.

\begin{theorem}{Adequacy}\label{theorem-adequacy}\\
  Let $M$ be well-typed term of \lamadio. We write $\values{M} := \{
  T = \ \parallel_i V_i \mid M \to^\ast T \}$. Similarly, for a net $\net R$
  with normal form $\net N$, we define $\values{\net R} = \{ \net S \mid \net N =
  \net S + \net S',\ \net S \text{ is a value net } \}$. Then
  $$\values{M} = \values{M^\bullet}$$
\end{theorem}

We first need to prove adequacy between {\lamadio} and {\lthis}:

\begin{theorem}{Adequacy for {\lthis}}\label{theorem-adequacy-lthis}\\
  Let $P$ be a term of {\lamadio} and $M = \widetilde{P}$ its translation in
  \lthis, if $M \to^\ast \sum_i M_i$ a normal form, then $\forall
  i, P \to^\ast P_i$ such that $M_i \rincl \widetilde{P_i}$. In other words, any
  term appearing in the normal form of the translation of $P$ is bounded by the
  translation of a reduct of $P$.
\end{theorem}

In particular, applied to values, this gives the sought property for {\lthis}:

\begin{corollary}\label{corollary-adequacy-simulation-orig}
  Let $P$ be a term of {\lamadio}, $M = \widetilde{P}$ its translation in
  {\lthis} such that $M \to^\ast M' + \mathbf{M''}$ where $M'$ is a normal
  form.
  \begin{itemize}
    \item If $M' = V \parallel N$, then $P \to^\ast U \parallel Q$ with $V =
      \widetilde{U}$.
    \item In particular, if $M' = \parallel_i V_i$, then $P \to^\ast \parallel_i
      U_i$ with $V_i = \widetilde{U}_i$
  \end{itemize}
\end{corollary}

The only problematic case is the non deterministic reduction
\rrulet{subst-r}{get} which creates new summands. The proof consist in showing
that these summands are actually limited in what they can do. Formally, they are
bounded by the initial term that is being reduced, in the sense of the preorder
$\rincl$ defined in~\cite{HamdaouiValiron2018}. The following lemma state that
indeed the case of deterministic reduction is trivial:

\begin{lemma}{Values preservation}\label{lemma-preservation-values}\\
  Let $M$ be a term of {\lamadio}, and $M \to^\ast M'$ without using
  \rrulet{subst-r}{get}. We define $\text{NF}(M) = \{ T \mid T \text{ normal}
  \text{ and } M \to^\ast T \}$. Then $\text{NF}(M) = \text{NF}(M')$
\end{lemma}

\begin{proof}
  If we do not use \rrulet{subst-r}{get}, the two reduction (full and
  non-deterministic) coincide and we use the confluence in {\lamadio}.
\end{proof}

\begin{lemma}{}\label{lemma-onestep-adequacy} Let $M = \widetilde{P}$ be the
  translation of a {\lamadio} term, such that $M \to M'$. Then there exists
  $M'',N'$, such that $M' \to^\ast M'' \rincl \widetilde{N'}$ and $N
  \to^{\{0,1\}} N'$, with all reductions from $M'$ to $M''$ not being
\rrulet{subst-r}{get}.  \end{lemma}

\begin{proof}
  The only redexes in $M$ are either premises of \rrulet{subst-r'}{},
  $(\beta_v)$ or \rrulet{subst-r}{get}.
  In the first case, it corresponds to a reducible set whose reduction can be
  carried on in $M'$ by pushing the upward substitution to the top and pushing
  down the corresponding generated downward substitutions, to obtain the translation of
  $N'$ (which is $N$ where we the set is reduced). We can proceed the same way
  with $(\beta_v)$ : this corresponds to a $\beta$-redex in $N$, where $N'$ is
  the result of reducing it, and $M''$ is obtained by pushing down the generated
  substitutions (variable and refenreces).

  Finally, if the reduction rule is \rrulet{subst-r}{get}, then either it
  choosed one of the available values and it corresponds exactly to a get
  reduction $N \to N'$, or it threw away available values, in which case $N =
  N'$, $M' = M''$ and clearly $M'' \rincl \widetilde{N'}$.
\end{proof}

\begin{proof}{Theorem~\ref{theorem-adequacy-lthis}}\\
  We proceed by induction on $\eta (M)$, the length of the longest reduction
  starting from $M$. If $\eta (M) = 0$, ie $M$ is a normal form, this is
  trivially true.
  \\
  To prove the induction step, consider a reduct $M'$ of $M$. We use
  \ref{lemma-onestep-adequacy} to get $M' \to^\ast M'' \rincl \widetilde{P'}$ for some reduct
  $P'$ of $P$, such that the reduction to $M''$ doesn't use
  \rrulet{subst-r}{get}. Thus, by \ref{lemma-preservation-values},
  $\text{NF}(M'') = \text{NF}(M')$. By induction on $M''$, $\forall T \in
  \text{NF}(M''), \exists Q, P \to^\ast Q \text{ and } T \rincl \widetilde{Q}$.
  But this is true for any reduct $M'$, and we have $\text{NF}(M) = \bigcup_{M \to
  M'} \text{NF}(M')$, hence this is true for $M$.
\end{proof}

From there, we get the result combining Theorem~\ref{theorem-adequacy-lthis} and
Theorem~\ref{theorem-nets-simulation-lthis}. One can extract from the proof of
Lemma~\ref{lemma-nets-termination-nfsep} that the translation of a summand that
is not a parallel of values either reduces to $\mathbf{0}$, or to a net which is not a
translation of a value.

\section{Conclusion}
\label{section:conclusion}

In this paper, we presented a translation of a $\lambda$-calculus with higher
order references and concurrency inside a fragment of differential proof nets.
While several translations of effectful languages have been proposed in the
literature, none supports this combination of features to our knowledge. We
introduced a generalization of communication areas, routing areas, which turned
out to be a useful device to encode references. More generally, we think that
routing areas can be used to express various kind of non-deterministic,
concurrent communications.

Modeling concurrency comes at a price, as terms such as $(\lambda x.
\set{r}{x})\ \get{r} \parallel (\lambda x. \set{r}{x})\ \get{r})$ translates to
a net that do not respect the differential proofs nets correctness
criterion~\cite{ehrhard:hal-00150274}. Inside each thread, the $\set{r}{x}$
depends on the $\get{r}$ whose value will replace the $x$ variable. But of these
$\get{r}$ may also depend from the $\set{r}{x}$ of the other thread, creating a
seemingly circular dependence which breaks the acyclicity required for
correctness. This ambiguity is avoided at execution, as the first $\get{r}$ to
reduce will be forced to chose an available {\setw}: these set-get dependencies
are in fact mutually exclusive. But proof nets seems unable to express this
subtlety. Differential {\LL} seems to suffer from more fundamental limitations
as a model of concurrency as pointed out by Mazza~\cite{Mazza:TrueConcMSCS}. It is
yet to be clarified how these results apply to the fragment presented here.
While seriously limiting what can be modeled in proof nets, this might not be an
obstacle when aiming for practical parallel or distributed implementations.

The enrichment of the source language, such as switching to a more realistic
erase-on-write semantic for stores, or the addition of new effects and features
(synchronization operations, sum types, divergence either by fixpoint or
references, \emph{etc .}) is the main focus of future work.

\bibliographystyle{plainurl}
\bibliography{info}
\end{document}